\documentclass[12pt]{article}

\usepackage{fullpage}
\usepackage{hyperref}
\usepackage[numbers]{natbib}
\renewcommand{\citet}[1]{\citeauthor{#1}~\citep{#1}}

\usepackage{xspace}

\usepackage{amssymb,amsmath,amsthm}

\newcommand{\email}[1]{\href{mailto:#1}{\nolinkurl{#1}}}

\newtheorem{theorem}{Theorem}
\newtheorem{corollary}{Corollary}
\newtheorem{lemma}{Lemma}

\usepackage{xcolor}
\usepackage{graphicx}

\newcommand{\secref}[1]{Section~\ref{#1}}
\newcommand{\thmref}[1]{Theorem~\ref{#1}}
\newcommand{\lemref}[1]{Lemma~\ref{#1}}
\newcommand{\figref}[1]{Figure~\ref{#1}}
\newcommand{\appref}[1]{Appendix~\ref{#1}}
\newcommand{\corref}[1]{Corollary~\ref{#1}}

\usepackage{textcase}

\newcommand{\tagthis}{\stepcounter{equation}\tag{\theequation}}

\newcommand{\ie}{i.e.,\xspace}
\newcommand{\eg}{e.g.,\xspace}

\newcommand{\R}{\mathbb{R}}
\DeclareMathOperator*{\E}{\mathbb{E}}
\newcommand{\midd}{:}
\newcommand{\rgeq}[1]{\R_{\smash\geq}^{#1}}
\newcommand{\rgeqk}{\rgeq{k}}

\newcommand{\vmax}{\bar{v}}

\newcommand{\aGSP}{$\alpha$-GSP\xspace}

\newcommand{\aVCG}{$\alpha$-VCG\xspace}
\newcommand{\aGFP}{$\alpha$-GFP\xspace}

\newcommand{\alphenumi}{\renewcommand{\theenumi}{({\alph{enumi}})}
    \renewcommand{\labelenumi}{\theenumi}}
		
\def\clap#1{\hbox to 0pt{\hss#1\hss}}
\def\mathllap{\mathpalette\mathllapinternal} \def\mathrlap{\mathpalette\mathrlapinternal} \def\mathclap{\mathpalette\mathclapinternal}
\def\mathllapinternal#1#2{\llap{$\mathsurround=0pt#1{#2}$}}
\def\mathrlapinternal#1#2{\rlap{$\mathsurround=0pt#1{#2}$}}
\def\mathclapinternal#1#2{\clap{$\mathsurround=0pt#1{#2}$}}

\def\clapstack#1{\mathclap{\substack{#1}}}

\begin{document}

\title{Non-Truthful Position Auctions Are More Robust to Misspecification\thanks{Valuable feedback from Dirk Bergemann, Allan Borodin, the anonymous referees, and seminar participants at Google, Facebook, Universit\"at Z\"urich, and the University of Glasgow is gratefully acknowledged. The second author was supported by the Einstein Foundation Berlin.}}

\author{%
	Paul D\"utting\thanks{%
		Department of Mathematics,
		London School of Economics,
		Houghton Street,
		London WC2A 2AE, UK.
		Email: \email{p.d.duetting@lse.ac.uk}.}
	\and 
	Felix Fischer\thanks{%
		School of Mathematical Sciences,
		Queen Mary University of London,
		Mile End Road,
		London E1 4NS, UK.
		Email: \email{felix.fischer@qmul.ac.uk}.}
	\and
	David C.~Parkes\thanks{%
		John A.~Paulson School of Engineering and Applied Sciences,
		Harvard University,
		33 Oxford Street,
		Cambridge, MA 02138, USA.
		Email: \email{parkes@eecs.harvard.edu}.}
}

\date{}

\maketitle

\begin{abstract}
	In the standard single-dimensional model of position auctions, bidders agree on the relative values of the positions and each of them submits a single bid that is interpreted in terms of these values. Motivated by current practice in sponsored search we consider a situation where the auctioneer uses estimates of the relative values, which may be imprecise, and show that under both complete and incomplete information a non-truthful mechanism is able to support an efficient outcome in equilibrium for a wider range of these estimates than the VCG mechanism. We thus exhibit a property of the VCG mechanism that may help explain the surprising rarity with which it is used even in settings with unit demand, a relative lack of robustness to misspecification of the bidding language. The result for complete information concerns the generalized second-price mechanism and lends additional theoretical support to the use of this mechanism in practice. Particularly interesting from a technical perspective is the result for incomplete information, which is driven by a surprising connection between equilibrium bids in the VCG mechanism and the generalized first-price mechanism.
	
	\medskip
	
	
	
	
	
	\noindent\emph{Keywords:} Position auctions; Sponsored search; Model misspecification; Robustness
	\end{abstract}


\section{Introduction}

The Vickrey-Clarke-Groves (VCG) mechanism stands as one of the pillars of mechanism design theory, but in the real world is used with surprising rarity. Past work has attributed this mismatch between theory and practice to a number of properties that affect the mechanism in certain settings, like susceptibility to collusion or prohibitive computational costs~\citep{AuMi06a,Roth07a}. Here we identify another property of the mechanism that may be problematic in practice, a relative lack of robustness to misspecification of the bidding language. Unlike most of the known deficiencies it applies already in settings with unit demand.


\subsection{The Model}

Our point of departure is the standard position auction model of \citet{EOS07a} and \citet{Vari07a}, where~$n$ bidders compete for the assignment of~$k$ positions. Bidders have unit demand and the valuation of bidder~$i$ for position~$j$ is given by $\beta_j\cdot v_i$, where $v_i$ is a bidder-specific value and a non-increasing vector $\beta=(\beta_1,\dots,\beta_k)$ 
describes the relative values of the positions. The most prominent application of this model is to sponsored search, which contributed a significant fraction to Google's advertising revenue of around \$95 billion in~2017~\citep{abc}. What is rather curious is that the most successful sponsored search services to date have used non-truthful auction mechanisms rather than the VCG mechanism.
GoTo, later re-branded as Overture and acquired by Yahoo, was the first company to provide such a service and used a generalized first-price (GFP) mechanism. Google and Microsoft use a generalized second-price (GSP) mechanism.
Facebook does use the VCG mechanism to place ads, but not in the context of sponsored search and not in a position auction.\footnote{The ad selection and pricing problem faced by Facebook is indeed very different from a position auction, as ads come in different formats and can be placed flexibly within the news feed or next to it. The situation is similar for contextual ads in sponsored search~\citep{VaHa14a}.}
It is hard to say in retrospect what led to the selection of the non-truthful mechanisms, and changing the mechanism at this point would clearly come with huge financial risks. We will see, however, that under certain assumptions choosing the non-truthful mechanisms may have been wise even if it was not entirely deliberate.

\enlargethispage{1\baselineskip}

\citeauthor{EOS07a}, \citeauthor{Vari07a}, and much of the subsequent literature have considered mechanisms that from each bidder~$i$ solicit a single bid~$b_i$, which is then interpreted as a vector of bids $\beta_j\cdot b_i$, one for each position~$j$, to determine a one-to-one assignment of bidders to positions and a monetary payment for each bidder. The conflation of a vector of bids into a single bid mirrors practical mechanisms~\citep[\eg][]{Bing15a,Goog15a}. It is a deliberate design decision that greatly increases ease of use for the bidders and has also been shown across mechanisms to eliminate certain outcomes that are undesirable from the auctioneer's point of view~\citep{Milg10a}. 

Any real-world use of a mechanism requires a choice of bidding language and bears a risk of misspecification.
 In the case of sponsored search the assumption that may be problematic is that the auctioneer knows the exact value of~$\beta$, which in practice is inferred using techniques from statistical machine learning~\cite{GraepelCBH10,McMahanHSYEGNPDGCLWHBK13}. While these techniques are capable of producing fairly accurate estimates, they will typically not produce an estimate that is entirely exact.%
 \footnote{It is worth pointing out here that search engines do observe when an ad is clicked. Thus, if the relative values of the positions were exactly proportional to the relative number of clicks, all auctions we consider could be implemented without any knowledge of~$\beta$ \citep{VaHa14a}. 
 There are, however, good reasons why the value of a position may depend on other factors besides the probability of a click. \citet{Milg10a} for example considers a situation with two types of users of a search engine, one of them genuinely interested in the products being advertised and one merely curious. The two types come with different rates at which clicks on an advertisement result in a purchase, but are indistinguishable from the point of view of the search engine.} 
The result will be a slight misspecification of the bidding language, such that the way in which a single bid is extrapolated to bids on individual positions is not completely aligned with the actual relative values of the positions. Our goal will be to analyze how robust different mechanisms are to this kind of misspecification.

To this end, we consider a generalization of the standard model in which mechanisms work with bids $\alpha_j\cdot b_i$, 
where $\alpha=(\alpha_1,\dots,\alpha_k)$ is a non-increasing vector and generally $\alpha\neq\beta$. We then compare the performance of different mechanisms as $\alpha$ and~$\beta$ vary, and ask specifically for which values of~$\alpha$ and~$\beta$ each mechanism possesses an efficient equilibrium.
We follow the usual approach of analyzing bidder behavior via game-theoretic reasoning, under the assumption that each bidder tries to maximize the value of the position it is assigned minus its payment.\footnote{Implicit in this approach is the assumption that a bidder can evaluate its utility and hence the relative values of the positions. That bidders but not the auctioneer are aware of the relative values seems reasonable in sponsored search, 
because the generation of value for an advertiser could be observable by that advertiser but not by the search engine. 
In the model of \citet{Milg10a} involving two types of users of a search engine, for example, advertisers would be able to observe purchase decisions on their websites but could be wary of sharing this information with the search engine due to trust concerns.}
In addition to the standard model of auction theory, where bidders only have incomplete information about one another's valuations, it has become common to analyze position auctions also in a complete-information model where valuations are common knowledge among the bidders. This is motivated by practical auctions that provide bidders with aggregate statistics of others' bids and thus enable best-response bidding, by empirical support that has been given for a family of Nash equilibria~\citep{Vari07a,EOS07a}, and by cyclic bidding behavior observed in the absence of pure Nash equilibria~\citep{EOS07a}. 
Like the vast majority of work on position auctions we focus on the maximization of welfare, which in practice can be seen as a maximization of customer satisfaction to ensure long-term success.
Extending our results to the maximization of revenue nevertheless is an interesting direction for future work.

\subsection{Results}

We show that under both complete and incomplete information, a non-truthful mechanism is able to support an efficient outcome in equilibrium for a strictly larger set of values of~$\alpha$ and~$\beta$ than the VCG mechanism. Failure of the VCG mechanism to produce an efficient outcome can in fact occur already when~$\alpha$ is very close to~$\beta$.

The result for complete information follows from a comparison with the GSP mechanism and is obtained in \secref{sec:complete} via the combination of two results. Existence of an efficient equilibrium in the VCG mechanism is first shown to imply the existence of an efficient outcome satisfying the stronger property of envy-freeness. Since envy-freeness is independent of the underlying mechanism, it then suffices to show that the GSP mechanism is able to produce an outcome with the same assignment and payments.

In \secref{sec:incomplete} the same type of result is shown to hold under incomplete information, but here the VCG mechanism is compared to the GFP mechanism and a more elaborate argument is required to establish superiority of the latter. We begin by using a standard technique for equilibrium characterization that equates the expected payments in an efficient equilibrium as given by Myerson's Lemma with the respective payments in the two mechanisms. This gives us a candidate equilibrium bidding function for each of the two mechanisms, and each of these functions constitutes an equilibrium if and only if it is strictly increasing almost everywhere. In the case of the VCG mechanism we encounter an ordinary differential equation, which we solve by appealing to a combinatorial equivalence. Even with the bidding functions for the VCG and GFP mechanisms at hand it is not trivial to show that the latter is increasing for a larger set of values of~$\alpha$ and~$\beta$, and we exploit a surprising connection between the two functions to show that this is indeed the case.

Strict superiority of the non-truthful mechanisms, and potential failure of the VCG mechanism when~$\alpha$ is very close to~$\beta$, is shown in Sections~\ref{sec:complete-example} and~\ref{sec:incomplete-example} by means of two examples. In these examples efficient equilibria cease to exist when mechanisms underestimate the value of less valuable positions. This makes the less valuable positions more attractive by reducing their associated payments, incentivizing bidders to shade their bid and do so more strongly as their value increases. The relatively lower payments in the VCG mechanism, for a fixed profile of bids, only magnify this effect and cause it to fail for smaller discrepancies between~$\alpha$ and~$\beta$. In settings with sufficiently many positions the failure can also occur when mechanisms overestimate the quality of some positions but underestimate the quality of others.
More generally, the GSP and GFP mechanisms seem to benefit from the relative simplicity of their payments, which for a given position depend only on one bid. In the VCG mechanism a particular bid may simultaneously affect the payments of many bidders, setting the correct equilibrium payments thus becomes impossible more quickly as~$\alpha$ and~$\beta$ move out of alignment.\footnote{An orthogonal requirement for equilibrium existence that favors different non-truthful mechanisms under complete and incomplete information is that a bidder's ability to control its own payment must match the degree of knowledge it has of other bidders' valuations. It is well known, for example, that the GFP mechanism may not possess a pure Nash equilibrium under complete information~\citep{EOS07a} and that the GSP mechanisms may not possess an efficient Bayes-Nash equilibrium under incomplete information~\cite{GoSw09a}, even when~$\alpha=\beta$.}

The focus of our analysis is on standard mechanisms for position auctions and their robustness to misspecification regarding the relative values of the positions. The investigation of additional mechanisms and parameters, of more general settings, and of the interaction between auction mechanisms and learning algorithms used to infer the parameters, provide ample scope for future work.

\subsection{Related Work}

The risk of misspecification inherent in the choice of a bidding language recalls a common aphorism in statistics, first formulated in this form by George Box~\citep{Box76a}, that ``all models are wrong.'' To Box, the interesting question was not whether a model is an exact representation of the real world, but whether it is close enough to the truth to be useful. The role of model misspecification in mechanism design was recently highlighted by \citet{MaPr17a}, who designed a mechanism for the screening problem whose performance deteriorates gracefully in the distance between assumed and true preferences. We consider instead the effect of misspecification of the bidding language of an auction, and ask how much imprecision different mechanisms can tolerate while still producing an efficient outcome.


An increased robustness of non-truthful mechanisms for position auctions in the sense we discuss here was first suggested by \citet{Milg10a} and investigated further by~\citet{DFP11a}.
The effect of misspecification of the bidding language on equilibria of the GSP mechanisms was studied also by \citet{AGV07a} and by \citet{BHN08a}.
What dis\-tinguishes our results from this past work is that we elucidate, for any given value of~$\beta$, which values of~$\alpha$ enable the existence of an efficient equilibrium in different mechanisms that use $\alpha$ as an estimate of $\beta$.
They thus apply to mechanisms currently in use, and allow us to draw conclusions about the relative robustness of these mechanisms to misspecification.

The performance of the VCG mechanism and that of alternative, non-truthful mechanisms has already been compared in the standard position auction model, where $\alpha=\beta$, and a number of authors have noted certain advantages or lack of disadvantages of the alternative mechanisms. Under complete information the GSP mechanism obtains the truthful VCG outcome in a locally envy-free equilibrium, and payments that at least match those of the truthful VCG outcome in any locally envy-free equilibrium~\citep{EOS07a,Vari07a}. While the former is true for a whole class of mechanisms that rank bidders in order of their bids, the GSP mechanism is the simplest mechanism within this class~\citep{BaRo10a}. Under incomplete information the GFP mechanism admits a unique Bayes-Nash equilibrium, which in expectation yields the truthful VCG outcome~\citep{ChHa13a}. Each of the two mechanisms has severe disadvantages in the respective other setting, such as non-existence of a pure Nash equilibrium or of an efficient Bayes-Nash equilibrium~\citep{EOS07a,GoSw09a}. In cases where equilibria exist, however, the worst-case welfare loss is bounded in the sense of a small price of anarchy~\citep{CKK+15a}.\footnote{In the language of the price-of-anarchy literature we essentially seek to characterize those values of~$\alpha$ and~$\beta$ for which the price of stability is one. Arguments similar to the ones used to establish the price-of-anarchy guarantees for the standard model also apply to the more general setting we study here, providing welfare guarantees that degrade gracefully in~$\alpha$ and~$\beta$.} When $\alpha=\beta$, and other things being equal, the VCG mechanism of course has the advantage of truthfulness across complete- and incomplete-information environments.

Our work highlights one advantage of non-truthful mechanisms for position auctions. A concurrent line of work has identified additional advantages of non-revelation mechanisms, 
such as amenability to statistical inference~\citep{ChHa14} and guaranteed revenue in dynamic settings~\citep{HJW13a} and across complete- and incomplete-information environments~\citep{DFP17a}.

Our results finally fit more generally into an increasing body of work that emphasizes robustness and simplicity in economic and algorithmic design. Relevant examples of this type of work in economics include a theory of mechanisms with more robust knowledge assumptions~\citep{BeMo05a,BBM16a,Carr17a} 
and the use of a greedy mechanism in the FCC Incentive Auction to achieve computational and strategic simplicity~\citep{MiSe14,DGR16a,DRT17a}. Additional examples come from algorithmic game theory, where recent work has obtained simple mechanisms with near-optimal revenue~\citep{Hartline09,CHMS10a,Babaioff14} or welfare~\citep{Christ08,Bhawa11,Feld13,DHS13,FGL15a,DFKL17a}, but has also pointed out computational barriers to near-optimal equilibria~\citep{Cai14,Rough14}.


\section{Preliminaries}



	We study the standard setting of position auctions with~$k$ \emph{positions} ordered by quality and~$n\geq k$ \emph{bidders} with \emph{unit demand} and \emph{single-dimensional valuations} for the positions.\footnote{The assumption that $n\geq k$ is without loss of generality, as in a setting with $n<k$ bidders the $k-n$ lowest-valued positions would never be assigned by the mechanisms we consider.} Denote by $\rgeqk=\{x\in\R^k\midd x_j>0,\text{$x_{j}\geq x_{j'}$ if $j<j'$}\}$ the set of $k$-dimensional vectors whose entries are positive and non-increasing. Given $\beta\in\rgeqk$, which we assume to be common knowledge among the bidders, the valuation of a particular bidder~$i$ can then be represented by a scalar $v_i\in\R$, such that $\beta_j v_i\geq 0$ is the bidder's value for position~$j$. We will use the notational convention that $\beta_j=0$ when $j>k$.
	
	A \emph{mechanism} in this setting receives a profile $b\in\R^n$ of bids, assigns positions to bidders in a one-to-one fashion, and charges each bidder a non-negative payment. It can be represented by a pair $(g,p)$ of an \emph{allocation rule} $g:\R^n\rightarrow S_n$ and a \emph{payment rule} $p:\R^n\rightarrow\R^n$, such that for each $i\in\{1,\dots,n\}$, $g_i(b)=j$ for $j\in\{1,\dots,k\}$ means that bidder~$i$ is assigned position~$j$ and $p_i(b)$ is the payment charged to bidder~$i$. We will be concerned exclusively with mechanisms that assign positions in non-increasing order of bids, and henceforth denote by~$g$ an allocation rule that does so and breaks ties in an arbitrary but consistent manner.
	The role of payments is to incentivize bids resulting in an \emph{efficient assignment}, \ie one where positions are assigned in order of valuations and \emph{social welfare} $\sum_{i=1}^n\beta_{g_i(b)}v_i$ is maximized.

	In reasoning about strategic behavior we make the usual assumption of quasi-linear preferences and consider two different models of information regarding the preferences of other bidders. Under quasi-linear preferences, the \emph{utility} $u_i(b,v_i)$ of bidder~$i$ with value~$v_i$, in a given mechanism and for a given bid profile~$b$, is equal to its valuation for the position it is assigned minus its payment, \ie $u_i(b,v_i)=\beta_{g_i(b)}v_i-p_i(b)$. In the \emph{complete information} model the values $v_i$ are common knowledge among the bidders. A bid profile $b$ is a \emph{Nash equilibrium} of a given mechanism if no bidder has an incentive to change its bid assuming that the other bidders don't change their bids, \ie if for every $i\in N$, 
\[
	u_i\bigl(b,v_i\bigr) = \max_{x\in\R} u_i\bigl((b_{-i},x),v_i\bigr),
\]
where $(b_{-i},x)=(b_1,\dots,b_{i-1},x,b_{i+1},\dots,b_n)$. A Nash equilibrium $b$ is \emph{efficient} if for all $i,j\in\{1,\dots,n\}$, $b_i>b_j$ whenever $v_i>v_j$.

In the \emph{incomplete information} model values $v_i$ are drawn independently from a continuous distribution with density function~$f$, cumulative distribution function $F$, and support $[0,\bar{v}]$ for some finite $\bar{v}\in\R_+$ we assume to be common knowledge among the bidders.\footnote{An analytical characterization of equilibria in the case of non-identical distributions is, unfortunately, well beyond the state of the art even for very simple settings. For a single item and two bidders with values drawn uniformly from distinct intervals, for example, this question was posed by \citet{Vick61a} and answered only recently, almost half a century later, by \citet{Kaplan12}.}
Our results in addition require existence and boundedness of the first three derivatives of~$F$. 
Since valuations are independent and identically distributed, an efficient assignment for all value profiles can only be obtained from a symmetric profile $(b,\dots,b)$ for some bidding function $b:\R\rightarrow\R$. 
The quantity of interest for strategic considerations under incomplete information is the \emph{expected utility} $u_i^b(x,v_i)$ of bidder~$i$ with value $v_i$ given that it bids $x\in\R$ and all other bidders use bidding function $b$, which is given by
\[
	\textstyle u_i^b(x,v_i) = \E_{v_j\sim F, j\neq i} \Bigl[ u_i\Bigl(v_i,\bigl(b(v_1),\dots,b(v_{i-1}),x,b(v_{i+1}),\dots,b(v_n)\bigr)\Bigr) \Bigr] .
\]
Bidding function~$b$ then is a \emph{Bayes-Nash equilibrium} if no bidder has an incentive to change its bid, \ie if for all $i\in\{1,\dots,n\}$ and $v_i\in[0,\bar{v}]$,
\begin{align}
	u_i^b(b(v_i),v_i) = \max_{x\in\R} u_i^b(x,v_i) .  \label{eq:nash}
\end{align}
A Bayes-Nash equilibrium $b$ is \emph{efficient} if it is increasing almost everywhere.

A mechanism that achieves efficiency in both Nash and Bayes-Nash equilibrium is the Vickrey-Clarke-Groves (VCG) mechanism. It uses allocation rule~$g$ and a payment rule~$p^\beta$ that charges each bidder its externality on the other bidders, which is equal to the additional utility bidders assigned lower positions would obtain by moving up one position. Denoting by $b_{(i)}$ the $(n-i+1)$st order statistic of $b$, such that $b_{(1)}\geq\dots\geq b_{(n)}$, and using the convention that $b_{(i)}=0$ when $i>n$, \vspace*{-1ex}
\[
	p^\beta_i(b) 
	= \sum_{\mathclap{j=g_i(b)}}^{k}(\beta_j-\beta_{j+1}) b_{(j+1)} .
\]
It is well known and not difficult to see that the VCG mechanism makes it optimal for each bidder to bid its true valuation irrespective of the bids of others, which is a stronger property than those required of a Nash or Bayes-Nash equilibrium. The resulting assignment is efficient. The resulting outcome of assignment and payments is in fact the bidder-optimal core outcome, and we will refer to it by that name.

While computation of payments in the VCG mechanism requires knowledge of the vector~$\beta$ of relative values, we will be interested instead in the ability of mechanisms to support an efficient outcome in equilibrium when only an inaccurate estimate~$\alpha\in\rgeqk$ of~$\beta$ is available to the auctioneer. To this end we consider parameterized variants of the three mechanisms that have been used and studied most extensively: the \aVCG mechanism, the \aGFP mechanism, and the \aGSP mechanism. The three mechanisms all use allocation rule~$g$, and their payment rules $p^V\!\!$, $p^F\!\!$, and $p^S$ respectively charge a bidder its externality, its bid on the position it is assigned, and the next-lower bid on that position. Using the convention that $\alpha_{j}=0$ when $j>k$,
\vspace*{-2ex}
\begin{align*}
	p^V_i(b) &= \sum_{\mathclap{j=g_i(b)}}^{k}(\alpha_j-\alpha_{j+1}) b_{(j+1)},\\
	p^F_i(b) &= \alpha_{g_i(b)} b_i, \quad\text{and} \\
	p^S_i(b) &= \alpha_{g_i(b)} b_{(g_i(b)+1)}.
\end{align*}
We will sometimes drop superscripts when the mechanism we are referring to is clear from the context.

\section{Complete Information}
\label{sec:complete}

We begin our analysis with the complete-information case. Here, when $\alpha=\beta$, the \aVCG mechanism has a truthful equilibrium, the \aGSP mechanism has an equilibrium that yields the bidder-optimal core outcome~\citep{EOS07a,Vari07a}, and the \aGFP mechanism may not have any equilibrium~\citep{ChHa13a}. When $\alpha\neq\beta$ the \aVCG mechanism loses its truthfulness, and it makes sense to ask under what conditions the \aVCG mechanism and the \aGSP mechanism possess an efficient equilibrium. 
To build intuition we first look at the special case with three positions and three bidders, before moving on to the general case.

\subsection{Three Positions and Three Bidders}
\label{sec:complete-example}
In the special case, valuations are given by vectors $v\in\R^3$ and $\beta\in\mathbb{R}^3_{\ge}$ while mechanisms use a vector $\alpha\in\mathbb{R}^3_{\ge}$ that may differ from~$\beta$. Our goal will be to understand which combinations of~$\alpha$ and~$\beta$ allow for the existence of a bid profile $b\in\R^3$ that is an equilibrium and leads to an efficient assignment. 
Assuming without loss of generality that $v_1\geq v_2\geq v_3>0$, efficiency requires that
\begin{equation} \label{eq:ci_eff}
	b_1\geq b_2\geq b_3.
\end{equation}
For~$b$ to be an equilibrium, none of the bidders may benefit from raising or lowering their respective bid and being assigned a different position. For the \aVCG mechanism this means that
\begin{align}
	\beta_1 v_1 - (\alpha_1-\alpha_2) b_2 - (\alpha_2-\alpha_3) b_3 &\geq \beta_2 v_1 - (\alpha_2-\alpha_3) b_3 , \label{eq:vcg1} \\
	\beta_1 v_1 - (\alpha_1-\alpha_2) b_2 - (\alpha_2-\alpha_3) b_3 &\geq \beta_3 v_1 , \label{eq:vcg2} \\
	\beta_2 v_2 - (\alpha_2-\alpha_3) b_3 &\geq \beta_1 v_2 - (\alpha_1-\alpha_2) b_1 - (\alpha_2-\alpha_3) b_3, \label{eq:vcg3} \\
	\beta_2 v_2 - (\alpha_2-\alpha_3) b_3 &\geq \beta_3 v_2 , \label{eq:vcg4} \\
	\beta_3 v_3 &\geq \beta_1 v_3 - (\alpha_1-\alpha_2) b_1 - (\alpha_2-\alpha_3) b_2 , \label{eq:vcg5} \\
	\beta_3 v_3 &\geq \beta_2 v_3 - (\alpha_2-\alpha_3) b_2 . \label{eq:vcg6}
 \end{align}
There is no upper bound on~$b_1$ and no lower bound on~$b_3$ except $b_3\geq 0$, and setting~$b_1$ to a large value and $b_3=0$ satisfies \eqref{eq:ci_eff}, \eqref{eq:vcg3}, \eqref{eq:vcg4}, and \eqref{eq:vcg5}. With this choice of $b_3$, and since $\beta_2v_1\geq\beta_3v_1$, \eqref{eq:vcg2} is implied by \eqref{eq:vcg1}. The \aVCG mechanism thus possesses an efficient equilibrium if and only if there exists a bid~$b_2$ such that
\begin{align}
	&(\alpha_1-\alpha_2) b_2 \leq (\beta_1 - \beta_2) v_1 , \label{eq:vcga}\\ 
	&(\alpha_2-\alpha_3) b_2 \geq (\beta_2-\beta_3) v_3. \label{eq:vcgb}
\end{align}

For the \aGSP mechanism the equilibrium conditions require that
\begin{align}
	\beta_1 v_1 - \alpha_1 b_2 &\geq \beta_2 v_1 - \alpha_2 b_3 , \label{eq:gsp1} \\
	\beta_1 v_1 - \alpha_1 b_2 &\geq \beta_3 v_1 , \label{eq:gsp2} \\
	\beta_2 v_2 - \alpha_2 b_3 &\geq \beta_1 v_2 - \alpha_1 b_1, \label{eq:gsp3} \\
	\beta_2 v_2 - \alpha_2 b_3 &\geq \beta_3 v_2 , \label{eq:gsp4} \\
	\beta_3 v_3 &\geq \beta_1 v_3 - \alpha_1 b_1 , \label{eq:gsp5} \\
	\beta_3 v_3 &\geq \beta_2 v_3 - \alpha_2 b_2 . \label{eq:gsp6}
\end{align}
There is again no upper bound on~$b_1$, and setting~$b_1$ to a large value satisfies \eqref{eq:gsp3} and \eqref{eq:gsp5}. It is, moreover, not difficult to see that~\eqref{eq:gsp2} is implied by~\eqref{eq:gsp1} and~\eqref{eq:gsp4}: by \eqref{eq:gsp4}, $\alpha_2 b_3\leq(\beta_2-\beta_3)v_2$, so~\eqref{eq:gsp1} implies that $\beta_1 v_1 - \alpha_1 b_2 \geq \beta_2 v_1 - (\beta_2 - \beta_3) v_2$; since $v_1\geq v_2$, this in turn implies \eqref{eq:gsp2}. The \aGSP mechanism thus possesses an efficient equilibrium if and only if there exist bids~$b_2\geq b_3$ such that
\begin{align}
	\alpha_1 b_2 &\leq (\beta_1-\beta_2) v_1 + \alpha_2 b_3, \label{eq:gspa} \\
	\alpha_2 b_3 &\leq (\beta_2-\beta_3) v_2, \label{eq:gspb} \\
	\alpha_2 b_2 &\geq (\beta_2-\beta_3) v_3. \label{eq:gspc}
\end{align}

To see that the constraints for the \aGSP mechanism are generally weaker than those for the \aVCG mechanism, note that the former can be satisfied even under the additional restriction that $b_2=b_3$ if there exists a bid~$b_2$ such that
\begin{align}
	(\alpha_1-\alpha_2) b_2 &\leq (\beta_1-\beta_2) v_1, \label{eq:gspx} \\
	\alpha_2 b_2 &\geq (\beta_2-\beta_3) v_3. \label{eq:gspy}
\end{align}
Indeed, any such bid satisfies~\eqref{eq:gspa} and~\eqref{eq:gspc}, while the smallest such bid satisfies~\eqref{eq:gspb} as well. The claim now follows because~\eqref{eq:gspx} is identical to~\eqref{eq:vcga} and \eqref{eq:gspy} easier to satisfy than~\eqref{eq:vcgb}. The latter comparison will in fact be strict when $\alpha_3>0$.

When $\alpha_3=0$ there is no strict separation and the two mechanisms are in fact identical, but this is a viable design choice only in the absence of a fourth bidder, when the payment for the last position is always zero. When there is a fourth bidder, then for both mechanisms $\alpha_3>0$ becomes a necessary condition for the existence of an efficient equilibrium and the separation between the mechanisms is strict. A formal treatment of the case with three positions and four bidders is given in \appref{app:three_four}. This treatment also suggests that the analysis becomes significantly more difficult as the number of positions and bidders increases and can no longer be solved by a straightforward comparison of the respective equilibrium conditions.

To further illustrate the separation between the two mechanisms, note that when $\beta_i\neq\beta_j$ and $\alpha_i\neq\alpha_j$ for all $i \neq j$, \eqref{eq:vcga} and \eqref{eq:vcgb} can be satisfied if and only if
\begin{align*}
	&\alpha_2 \geq \frac{\alpha_1(\beta_2-\beta_3)v_3 + \alpha_3(\beta_1-\beta_2)v_1}{(\beta_1-\beta_2)v_1+(\beta_2-\beta_3)v_3}, \\
	\intertext{and~\eqref{eq:gspa},~\eqref{eq:gspb},~and \eqref{eq:gspc} if and only if}
	&\alpha_2 \geq \frac{\alpha_1(\beta_2-\beta_3)v_3}{(\beta_1-\beta_2)v_1+(\beta_2-\beta_3)v_3}.
\end{align*}
We compare these bounds in \figref{fig:complete33}. The bounds suggest that only an underestimation of~$\beta$ is problematic, while efficient equilibria are preserved by both mechanisms when $\alpha\geq\beta$. The analysis in \appref{app:three_four} shows that this, also, is an artifact of the case with three positions and three bidders and ceases to hold when there is an additional bidder.
\begin{figure}[tb]
\centering
\includegraphics{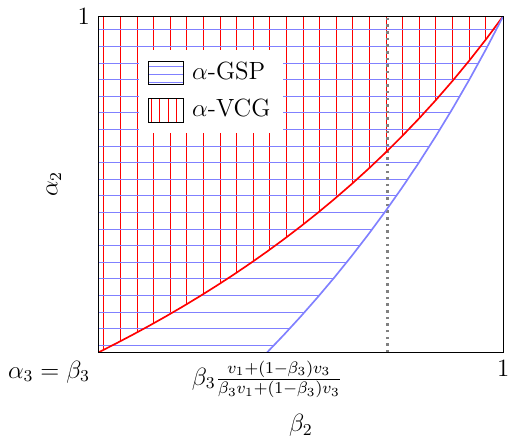}
\caption{Comparison of the \aGSP and \aVCG mechanisms under complete information, for a setting with three positions and three bidders where $\beta_1=\alpha_1=1$ and $0<\beta_3=\alpha_3<1$. The hatched areas indicate the combinations of $\alpha_2$ and $\beta_2$ for which the two mechanisms respectively possess an efficient equilibrium. The dotted line illustrates the performance of the mechanisms for a particular value of $\beta_2$. When $v_1=10$, $v_3=6$, and $\alpha_3=\beta_3=0.3$, this line would lie at $\beta_2=0.8$ and would intersect the curve for the \aGSP mechanism at $\alpha_2=0.6$ and that for the \aVCG mechanism at $\alpha_2=0.72$. Any point on the dotted line between these two intersection points corresponds to a value of $\alpha_2$ for which the \aGSP mechanism possesses an efficient equilibrium and the \aVCG mechanism does not.}
\label{fig:complete33}
\end{figure}

\subsection{The General Case}
\label{sec:complete-general}


We proceed to show that superiority of the \aGSP mechanism over the \aVCG mechanism in preserving efficient equilibria also holds in general. The following result establishes a weak superiority for arbitrary numbers of bidders and positions and arbitrary valuations. Examples in which only the \aGSP mechanism preserves an efficient equilibrium are straightforward to construct, and indeed we have already done so for a specific setting.
\begin{theorem}  \label{thm:complete}
	Let $\alpha,\beta\in\R_\ge^k$, $v\in\R^n$. Then the \aGSP mechanism possesses an efficient Nash equilibrium for valuations given by~$\beta$ and~$v$ whenever the \aVCG mechanism does.
\end{theorem}

Rather than by a direct comparison of the equilibrium conditions of the two mechanisms, as we have done in the special case, we prove the theorem by appealing to the stronger requirement of envy-freeness, which has played a significant role also in earlier work on VCG and GSP position auctions~\citep{Vari07a,EOS07a}. 
We first show, in \lemref{lem:envy-free}, that existence of an efficient Nash equilibrium in the \aVCG mechanism implies the existence of an efficient bid profile satisfying envy-freeness. Here, bid profile~$b$ is called envy-free if no bidder prefers a different position to the one it is currently assigned at the \emph{current} payment for the former, \ie if for all $i\in\{1,\dots,n\}$,
\begin{align}  \label{eq:ef}
	\beta_{g_i(b)}v_i-p_i(b) = \max\nolimits_{j\in\{1,\dots,n\}} \beta_{g_j(b)}v_i-p_j(b) .
\end{align}
For both the \aVCG mechanism and the \aGSP mechanism envy-freeness implies the equilibrium condition because the current payment is a lower bound on the actual payment of the bidder if by either mechanism it was assigned that position. Moreover, and in contrast to the equilibrium condition, envy-freeness can be viewed as a requirement that depends only on the allocation and payments and not on the underlying mechanism. We can thus complete the proof by establishing the existence of a mapping from bid profiles in the \aVCG mechanism to bid profiles in \aGSP mechanism that preserves allocation and payments, which we do in \lemref{lem:equilibrium}.

Assume without loss of generality that $v_1\geq v_2\geq\dots\geq v_n$ and that in an efficient allocation, for $1\leq i\leq\min\{n,k\}$, bidder $i$ is assigned position $i$.
Specializing and rearranging \eqref{eq:nash}, a bid profile~$b$ with $b_1 \geq\dots\geq b_n$ is a Nash equilibrium of the \aVCG mechanism if for all $i,j\in\{1,\dots,n\}$,
\begin{align}
(\alpha_j-\alpha_{j+1}) b_j &\geq (\beta_j-\beta_i) v_i - \sum_{\mathclap{t=j+1}}^{i-1}(\alpha_t-\alpha_{t+1})b_t & \text{if $j<i$,} \label{eq:vcg-eq-up} \\
(\alpha_i-\alpha_{i+1}) b_{i+1} &\leq (\beta_i-\beta_j)v_i - \sum_{\mathclap{t=i+1}}^{j-1} (\alpha_t-\alpha_{t+1})b_{t+1} & \text{if $j>i$.} \label{eq:vcg-down}
\end{align}
These conditions constrain the utility of bidder~$i$ if instead of position~$i$ it was assigned a position~$j$ that is respectively above or below~$i$. Note that in the latter case the payment of bidder~$i$ for position~$j$ is equal to the current payment for this position, where in the former case it may be higher.
Specializing~\eqref{eq:ef}, a bid profile~$b$ with $b_1 \geq\dots\geq b_n$ is envy-free if for all $i,j\in\{1,\dots,n\}$, in addition to~\eqref{eq:vcg-down} and instead of~\eqref{eq:vcg-eq-up},\footnote{Condition~\eqref{eq:vcg-ef-up} is, unlike~\eqref{eq:vcg-eq-up}, symmetric to~\eqref{eq:vcg-down}. \citet{Vari07a} therefore refers to bid profiles satisfying~\eqref{eq:vcg-ef-up} and~\eqref{eq:vcg-down} as \emph{symmetric equilibria}.}
%
%
\begin{align}
(\alpha_j-\alpha_{j+1}) b_{j+1} &\geq (\beta_j-\beta_i) v_i - \sum_{\mathclap{t=j+1}}^{i-1}(\alpha_t-\alpha_{t+1})b_{t+1}  \qquad\text{if $j<i$.} \label{eq:vcg-ef-up}
\end{align}

Envy-freeness is a stronger requirement than that of being an equilibrium, but we will see that it comes for free in the sense that existence of an efficient equilibrium automatically implies existence of an efficient equilibrium satisfying envy-freeness. 
\begin{lemma}  \label{lem:envy-free}  
	Let $\alpha,\beta\in\R_\ge^k$ and $v\in\R^n_\ge$, and assume that the \aVCG mechanism possesses an efficient equilibrium. Then the \aVCG mechanism possesses an efficient equilibrium that is envy-free.
\end{lemma}
\begin{proof}
	We will show existence and envy-freeness of a particular type of efficient equilibrium that we will call bidder-pessimal, in which each of the bids $b_2,\dots,b_n$ is maximal among all efficient equilibria.\footnote{In the case where $\alpha=\beta$, the truthful equilibrium of the \aVCG mechanism is bidder-\emph{optimal} among all envy-free outcomes, \ie its bids and payments are minimal~\citep{Leon83a}.}
	
	First note that~\eqref{eq:vcg-eq-up} only imposes lower bounds on the bids and remains satisfied when bids are increased.
	For~\eqref{eq:vcg-down}, the case where $j=i+1$ implies all other cases, because 
%
\begin{align*}
	(\beta_i-\beta_j)v_i - \sum_{\mathclap{t=i+1}}^{j-1}(\alpha_t-\alpha_{t+1})b_{t+1}
	&\geq (\beta_i-\beta_j)v_i - \sum_{\mathclap{t=i+1}}^{j-1}(\beta_t-\beta_{t+1})v_{t} \\
	&\geq (\beta_i-\beta_j)v_i - \sum_{\mathclap{t=i+1}}^{j-1}(\beta_t-\beta_{t+1})v_{i} \\
	&= (\beta_i-\beta_{i+1})v_i,
\end{align*}
where the first inequality holds because, by the fact that~$b$ is an efficient equilibrium and hence by~\eqref{eq:vcg-down}, $(\alpha_t-\alpha_{t+1})b_{t+1} \leq (\beta_t-\beta_{t+1})v_t$ when $i+1\leq t\leq j-1$, and the second inequality because, by efficiency, $v_t\leq v_i$ for all such $t$.  
	Bid~$b_i$, for $i\in\{2,\dots,n\}$, thus is subject to only two upper bounds on, $b_i\leq b_{i-1}$ by efficiency and $(\alpha_{i-1}-\alpha_i)b_i\leq(\beta_{i-1}-\beta_i)v_{i-1}$ by~\eqref{eq:vcg-down}. Increasing each of these bids as much as possible yields a bid profile~$b$ such for all $i\in\{2,\dots,n\}$,
	\begin{align}  \label{eq:ef_bids}
		b_i = \begin{cases}
			\min\left(b_{i-1},\frac{(\beta_{i-1}-\beta_i)v_{i-1}}{\alpha_{i-1}-\alpha_i}\right) & \text{if $\alpha_{i-1}\neq\alpha_i$}, \\
			b_{i-1} & \text{otherwise.}
		\end{cases}
	\end{align}
	

	We now claim that~$b$ satisfies \eqref{eq:vcg-ef-up} and begin by showing this for the special case where $j=i-1$, which requires that for all $i\in\{2,\dots,n\}$, 
	\begin{align}  \label{eq:ef_up_local}
		(\alpha_{i-1}-\alpha_i) b_i \geq (\beta_{i-1}-\beta_i)v_i.
	\end{align}
	By~\eqref{eq:ef_bids} it suffices to distinguish two cases.
	If $\alpha_{i-1}\neq\alpha_i$ and $b_i=(\beta_{i-1}-\beta_i)v_{i-1}/(\alpha_{i-1}-\alpha_i)$, then
	\[
		(\alpha_{i-1}-\alpha_i) b_i = (\beta_{i-1}-\beta_i)v_{i-1} \geq (\beta_{i-1}-\beta_i)v_i,
	\]
	where the inequality holds because $v_i\geq v_{i+1}$. 
	If instead $b_i = b_{i-1}$, then 
	\[
		(\alpha_{i-1}-\alpha_i) b_i = (\alpha_{i-1}-\alpha_i) b_{i-1} \geq (\beta_{i-1}-\beta_i)v_i,
	\]
	where the inequality holds by~\eqref{eq:vcg-eq-up}.
	
	For the general case let $i,j\in\{1,\dots,n\}$ with $j<i$. Then
	\begin{align*}
		(\beta_j-\beta_i)v_i - \sum_{t=j+1}^{i-1}(\alpha_t-\alpha_{t+1})b_{t+1}
		&\leq (\beta_j-\beta_i)v_i - \sum_{\mathclap{t=j+1}}^{i-1}(\beta_t-\beta_{t+1})v_{t+1} \\
		&\leq (\beta_j-\beta_i)v_i - \sum_{\mathclap{t=j+1}}^{i-1}(\beta_t-\beta_{t+1})v_i \\
		&= (\beta_j-\beta_{j+1})v_{i}\\
		&\leq (\beta_j-\beta_{j+1})v_{j+1}\\
		&\leq (\alpha_j-\alpha_{j+1}) b_{j+1},
	\end{align*}
	where the first and last inequality hold because, by \eqref{eq:ef_up_local}, $(\alpha_t-\alpha_{t+1})b_{t+1} \geq (\beta_t-\beta_{t+1})v_{t+1}$ for $t=j+1,\dots,i-1$ and $(\beta_j-\beta_{j+1})v_{j+1}\leq (\alpha_j-\alpha_{j+1}) b_{j+1}$, and the second and third inequality because $v_{t+1}\geq v_i$ when $t+1\leq i$ and $v_i\leq v_j$ when $j<i$.\footnote{In extending the claim from the special to the general case we have in fact shown that, subject to efficiency, \emph{local} envy-freeness with regard to the position directly above implies envy-freeness with regard to all higher positions. Similar results have appeared in prior work~\citep[\eg][Fact~5]{Vari07a}, but only for the case where $\alpha=\beta$.}
\end{proof}

We proceed to show that any bid profile in the \aVCG mechanism can be mapped to a bid profile that in the \aGSP mechanism yields the same allocation and payments. Applying this mapping to an efficient envy-free bid profile like the one identified by \lemref{lem:envy-free}, and noting that envy-freeness implies the equilibrium condition, then shows \thmref{thm:complete}.


\begin{lemma}  \label{lem:equilibrium}  
	Let $\alpha,\beta\in\R_\ge^k$ and $v,b\in\R^n_\ge$. Let $b^S\in\R^n$ such that for all $i\in\{1,\dots,n\}$,
	\begin{align*}
		b_{i}^S = \begin{cases}
		b_2^S & \text{if $i=1$,} \\[.6ex]
		\frac{p_{i-1}^V(b)}{\alpha_{i-1}} &\text{if $i\in\{2,\dots,k+1\}$ and $\alpha_{i-1}>0$}, \\[.5ex]
		0 & \text{otherwise.}
		\end{cases}
	\end{align*}
	Then $b^S_1\geq\dots\geq b^S_n$ and for all $i\in\{1,\dots,n\}$, $p_i^S(b^S)=p_i^V(b)$.
\end{lemma}
\begin{proof}
	Note that $b_1^S=b_2^S$. Let $j=\min\{i\midd\alpha_i=0\}$ and note that $b_i^S=0$ when $i>j$. For the first part of the claim it thus suffices to show that $b_i^S \geq b_{i+1}^S$ for $i=2,\dots,j-1$, which we do in two steps. First, for all $i\in\{2,\dots,j\}$, 
\begin{equation}  \label{eq:eff1}
	b_i^S 
	= \frac{p^V_{i-1}}{\alpha_{i-1}}
	= \frac{\sum_{t=i-1}^{k}(\alpha_t-\alpha_{t+1})b_{t+1}}{\alpha_{i-1}}
	\leq \frac{\sum_{t=i-1}^{k} (\alpha_t-\alpha_{t+1})b_i}{\alpha_{i-1}}
	= \frac{(\alpha_{i-1}-\alpha_{k+1})b_i}{\alpha_{i-1}} = b_i,
\end{equation}
where the first two equalities respectively hold by definition of $b_i^S$ and $p^V_{i-1}$, the inequality because $b_1\geq\dots\geq b_n$, and the last equality because, by convention,~$\alpha_{k+1}=0$. Then, for all $i\in\{2,\dots,j-1\}$,
\begin{align*}
	b_i^S 
	= \frac{p^V_{i-1}}{\alpha_{i-1}}
	&= \frac{(\alpha_{i-1}-\alpha_{i})b_{i}+p_i^V}{\alpha_{i-1}}\\
	&= \frac{(\alpha_{i-1}-\alpha_{i})b_{i}+\alpha_i b^S_{i+1}}{\alpha_{i-1}}
	\geq \frac{(\alpha_{i-1}-\alpha_{i})b^S_{i+1}+\alpha_i b^S_{i+1}}{\alpha_{i-1}}
	=b_{i+1}^S,
\end{align*}
where the first and third equalities hold by definition of $b_i^S$, the second equality exploits the recursive nature of the definition of $p^V_{i-1}$, and the inequality uses that $b_i\geq b_{i+1}$ and that, by~\eqref{eq:eff1}, $b_{i+1}\geq b^S_{i+1}$.

The second part of the claim is satisfied for $i<j$ because $p_i^S=\alpha_{i}b_{i+1}^S=\alpha_{i}p_i^V/\alpha_{i}=p_i^V$, and for $i\geq j$ because $\alpha_i=0$ for all $i\geq j$ and thus $p_i^S=p_i^V=0$.
\end{proof}

The above analysis in fact shows that any envy-free equilibrium of the \aVCG mechanism is preserved by the \aGSP mechanism. Since the bidder-optimal core outcome is envy-free~\citep{Leon83a}, we thus have the following.
\begin{corollary}  \label{cor:complete}
	Let $\alpha,\beta\in\R_\ge^k$, $v\in\R^n$. Then the \aGSP mechanism obtains the bidder-optimal core outcome in a Nash equilibrium for valuations given by~$\beta$ and $v$ whenever the \aVCG mechanism does.
\end{corollary}



\section{Incomplete Information}
\label{sec:incomplete}

We now turn to incomplete-information environments, where bidders only possess probabilistic information regarding one another's valuations. Here the \aGSP mechanism may fail to possess an efficient equilibrium even when $\alpha=\beta$.\footnote{\citet{GoSw09a} gave a characterization of those values of~$\alpha$ that enable equilibrium existence in this case. The result can be strengthened in our setting to show that for some values of~$\beta$ no choice of~$\alpha$ leads to an efficient equilibrium.} When $\alpha=\beta$, the \aVCG mechanism of course maintains its truthful dominant-strategy equilibrium. Another good mechanism in this case is the \aGFP mechanism, which differs from the \aGSP mechanism in its use of first-price rather than second-price payments. While sharing the latter's non-truthfulness it possesses a unique Bayes-Nash equilibrium for any value of $\alpha$, and this equilibrium yields the bidder-optimal core outcome~\citep{ChHa13a}.

Given these results it is quite natural to ask how successful the \aVCG and \aGFP mechanisms are in maintaining an efficient equilibrium outcome when $\alpha\neq\beta$. The answer to this question is strikingly similar to the complete-information case in that the non-truthful mechanism is again more robust, for arbitrary values of $\alpha$ and $\beta$ and independent and identically distributed valuations according to any distribution satisfying mild technical conditions.
Our analysis uses \citeauthor{Myer81a}'s classical characterization of possible equilibrium bids to identify, for either of the two mechanisms, conditions on $\alpha$ and $\beta$ that are necessary and sufficient for equilibrium existence. The conditions for the \aVCG mechanism turn out to be more demanding. Just as we did for complete-information environments, we begin by considering a special case, this time with two positions, three bidders, and valuations drawn uniformly at random from the unit interval. The special case is used to build intuition, and introduce the necessary machinery, for the general result.

\subsection{Two Positions and Three Bidders}
\label{sec:incomplete-example}

Let $v_1,v_2,v_3$ be drawn independently from the uniform distribution on $[0,1]$. Let $\alpha,\beta\in\R_{\ge}^2$ with $\alpha_2,\beta_2>0$, and assume without loss of generality that $\alpha_1=\beta_1=1$. Our goal will again be to characterize the values of $\alpha$ and $\beta$ for which given mechanisms of interest, in this case the \aGFP and \aVCG mechanisms, admit an efficient equilibrium. Behavior under incomplete information can be described by a profile of bidding functions, one for each bidder, that map the bidder's value to its bid. It is clear that in a symmetric setting like ours efficient outcomes can only result from symmetric bidding functions, so we will be interested in functions $b^F:\R\rightarrow\R$ that yield an efficient equilibrium in the \aGFP mechanism and functions $b^V:\R\rightarrow\R$ that achieve the same in the \aVCG mechanism.

The standard technique for equilibrium analysis under incomplete information uses a seminal result of \citeauthor{Myer81a} that characterizes the expected allocation and payments in equilibrium in terms of the allocation probabilities induced by a mechanism and bidders' bidding functions. The result was originally formulated for truthful mechanisms, but equivalent conditions exist for arbitrary bidding functions that instead of being in equilibrium provide a best response among values in their range. The latter is obviously a necessary condition for equilibrium, and can be turned into a sufficient condition by arguing that no better response exists outside the range. For our setting and notation we have the following result.
\begin{lemma}[\citet{Myer81a}]  \label{lem:myerson}
	Consider a position auction for an environment with~$n$ bidders,~$k$ positions, and $\beta\in\R^k_\ge$. Assume that bidders use a bidding function~$b$ with range~$X$, and that a bidder with value~$v$ is consequently assigned position $s\in\{1,\dots,k\}$ with probability $P_{s}(v)$. Then $u(b(v),v)=\max_{x\in X}u(x,v)$ for all $v\in[0,\vmax]$ if and only if the following holds:
	\begin{enumerate} \alphenumi
		\item the expected allocation $\sum_{s=1}^{k}P_s(v)\beta_s$ is non-decreasing in~$v$, and  \label{itm:monotone}
		\item the payment function $p$ satisfies  \label{itm:envelope}
		\begin{equation}  \label{eq:eff}
			\E[p(v)] = p(0) + \sum_{s=1}^{k} \beta_s \int_{0}^{v} \frac{d P_s(z)}{dz} z \,dz.
		\end{equation}
	\end{enumerate}
\end{lemma}
All mechanisms we consider set $p(0)=0$ and use an efficient allocation rule, for which
\begin{align*}
	P_s(v) &= \binom{n-1}{s-1}(1-F(v))^{s-1}(F(v))^{n-s}
\end{align*}
and~\ref{itm:monotone} is satisfied. Together with our assumptions on~$F$, efficiency mandates further that~$b$ must increase almost everywhere.


In the special case with two positions and three bidders with values distributed uniformly on the unit interval we have that $P_1(v)=F^2(v)=v^2$ and $P_2(v)=\binom{2}{1}F(v)(1-F(v))=2v(1-v)$, payments in any efficient equilibrium can thus be described by a function $p^E:\R\rightarrow\R$ satisfying
\begin{align}  \label{eq:ex-eff}
	\E[p^E(v)]
	&= \beta_1 \int_0^v \frac{dP_1(z)}{dz} z \,dz + \beta_2 \int_0^v \frac{dP_2(z)}{dz} z \,dz \notag \\
	&= \frac{2}{3}\beta_1v^3 + \beta_2 v^2 - \frac{4}{3}\beta_2 v^3. 
\end{align}

A candidate equilibrium bidding function for the \aGFP mechanism can now be obtained by writing the expected payment in terms of bidding function~$b^{F}$, equating the resulting expression with~\eqref{eq:ex-eff}, and solving for $b^{F}$.
In the \aGFP mechanism a bidder with value~$v$ that is allocated position~$s$ pays $\alpha_s b^{F}(v)$, its expected payment therefore satisfies
\begin{align}  \label{eq:ex-gfp}
	\E[p^{F}(v)] 
	&= P_1(v) \alpha_1 b^{F}(v) + P_2(v) \alpha_2 b^{F}(v) \notag \\
	&= (\alpha_1 v^2 + 2\alpha_2v-2\alpha_2v^2) b^{F}(v).
\end{align}
By \lemref{lem:myerson} the expressions in~\eqref{eq:ex-eff} and~\eqref{eq:ex-gfp} must be the same. Equating them yields
\begin{align*}  \label{eq:ex-gfp-cand}
	&b^{F}(v) = \frac{2/3 \cdot v^3 - 4/3 \cdot \beta_2 v^3 + \beta_2 v^2}{v^2-2\alpha_2v^2+2\alpha_2v}
\end{align*}
when $v>0$, and we can set $b^F(0)=0$ for convenience.\footnote{Application of l'Hospital's rule shows that $\lim_{v\rightarrow 0} b^{F}(v)=0$, so this choice makes $b^{F}$ increasing.}
Bidding below $b^F(0)=0$ is impossible, bidding above $b^F(\bar{v})$ is dominated,\footnote{Since equilibrium bidding functions must be increasing almost everywhere, bidding above $b^F(\bar{v})$ would not increase the probability of winning, and it would also not lead to a lower payment.} and~$b^F$ satisfies the second condition of \lemref{lem:myerson} by construction. 
The \aGFP mechanism thus has an efficient equilibrium if and only if~$b^F$ is increasing almost everywhere. 
Taking the derivative we obtain
\begin{equation*}
	\frac{db^F(v)}{dv} =
	\frac{(\frac{4}{3}v-\frac{8}{3}\beta_2v+\beta_2)(v-2\alpha_2v+2\alpha_2)}{(v-2\alpha_2v+2\alpha_2)^2} -
	\frac{(1-2\alpha_2)(\frac{2}{3}v^2-\frac{4}{3}\beta_2v^2+\beta_2v)}{(v-2\alpha_2v+2\alpha_2)^2}.
\end{equation*}
The sign of this expression is determined by the sign of its numerator, and it turns out that the numerator is positive at~$0$ and, depending on the value of $\beta_2$, either non-decreasing everywhere on $[0,1]$ or decreasing everywhere on $[0,1]$. Indeed, $db^F(v)/dv|_{v=0}=\beta_2/(2\alpha_2)>0$, and the derivative of the numerator, $(4/3-8/3\beta_2)(v-2\alpha_2v+2\alpha_2)$, is non-negative when $\beta_2\leq 1/2$ and negative when $\beta_2>1/2$. In the case where $\beta_2>1/2$ we need that
\begin{align*}
	&\left.\frac{db^F(v)}{dv} \right|_{\mathrlap{v=1}}\;\; = \left(\frac{4}{3}-\frac{5}{3} \beta_2\right) - (1-2\alpha_2) \left(\frac{2}{3}-\frac{1}{3}\beta_2\right) \geq 0,
\end{align*}
which holds when
\begin{equation*}  
	\alpha_2 \geq \frac{2 \beta_2-1}{2-\beta_2}.
\end{equation*}
We conclude that the \aGFP mechanism possesses an efficient equilibrium if and only if $\beta_2 \le 1/2$ or $\alpha_2 \ge (2 \beta_2-1)/(2-\beta_2).$

Analogously, in the \aVCG mechanism, the payment of a bidder with value~$v$ satisfies
\begin{align} \label{eq:ex-vcg}
	\E[p^V(v)]
	&= P_1(v) \biggl[ (\alpha_1 - \alpha_2) \int_0^v \frac{2t}{v^2} b^V(t) \,dt + 
	\alpha_2 \int_0^v \frac{2(v-t)}{v^2} b^V(t) \,dt \biggr]
	+ P_2(v) \alpha_2 \int_0^v \frac{1}{v} b^V(t) \,dt \notag \\
	&= (2\alpha_1-4\alpha_2)\int_0^v t b^V(t) \,dt + 2\alpha_2 \int_0^v b^V(t) \,dt,
\end{align}
where $2t/v^2=2F(t)f(t)/F(v)^2$ and $2(v-t)/v^2=2F(v-t)f(t)/F(v)^2$ are the densities of the second and third highest values given that the bidder's value~$v$ is the highest, and $1/v = f(t)/F(v)$ is the density of the third highest value given that~$v$ is the second highest. 
By \lemref{lem:myerson} the expressions in~\eqref{eq:ex-eff} and~\eqref{eq:ex-vcg} must again be the same. Taking the derivatives of both and solving for $b^{V}(v)$ yields
\begin{align*}  \label{eq:ex-vcg-cand}
	b^V(v) = \frac{2v^2-4\beta_2 v^2+2\beta_2 v}{2v-4\alpha_2v+2\alpha_2}
\end{align*}
when $v<1$, and we can extend $b^V$ appropriately when $v=1$.\footnote{We have assumed that $\alpha_2>0$, so the denominator vanishes only when $v=\alpha_2=1$. If $\beta_2 < 1$, then $\lim_{v\rightarrow 1}b^V(v)=\infty$. If $\beta_2=1$, application of l'Hospital's rule shows that $\lim_{v\rightarrow 1}b^V(v)=1$.}
By the same argument as before, the \aVCG mechanism has an efficient equilibrium if and only if $b^V$ is increasing almost everywhere. Taking the derivative we obtain
\begin{equation*}
	\frac{db^V(v)}{dv} =
	\frac{(4v-8\beta_2 v+2\beta_2)(2v-4\alpha_2v+2\alpha_2)}{(2v-4\alpha_2v+2\alpha_2)^2} -
	\frac{(2-4\alpha_2)(2v^2-4\beta_2 v^2+2\beta_2 v)}{(2v-4\alpha_2v+2\alpha_2)^2}.
\end{equation*}
When $\alpha_2<1$ the sign of this expression is determined by its numerator, which is positive at~$0$ and, depending on the value of~$\beta_2$, either non-decreasing everywhere on $[0,1]$ or decreasing everywhere on $[0,1]$. Indeed, $db^F(v)/dv|_{v=0}=\beta_2/\alpha_2>0$, and the derivative of the numerator, $(4-8\beta_2)(2v-4\alpha_2v+2\alpha_2)$, is non-negative when $\beta_2\leq 1/2$ and negative when $\beta_2>1/2$. When $\beta_2>1/2$ we need that
\begin{align*}
	&\left.\frac{db^V(v)}{dv} \right|_{\mathrlap{v=1}}\;\; = \frac{(4-6\beta_2)(2-2\alpha_2) - (2-4\alpha_2)(2-2\beta_2)}{(2-2\alpha_2)^2} \geq 0,
\end{align*}
which for $\alpha_2<1$ holds when
\begin{equation*}  
	\alpha_2 \geq 2-\frac{1}{\beta_2}.
\end{equation*}
When $\alpha_2=1$ the above reasoning still applies as long as $v<1$, so $b^V(v)$ is increasing almost everywhere when
\[
	\lim_{v\rightarrow 1}\frac{db^V(v)}{dv} \geq 0.
\]
This is indeed the case, as $\lim_{v\rightarrow 1}db^V(v)/dv=\infty$ when $\beta_2<1$, and $\lim_{v\rightarrow 1}db^V(v)/dv=1$ when $\beta_2=1$ by applying l'Hospital's rule twice. We conclude that the \aVCG mechanism possesses an efficient equilibrium if and only if $\beta_2 \le 1/2$ or $\alpha_2 \geq 2-1/\beta_2$.

It is now not hard to see that the equilibrium condition for the \aGFP mechanism is easier to satisfy than that for the \aVCG mechanism. In fact, for the \aVCG mechanism, efficient equilibria may cease to exist even when $\alpha_2$ is very close to $\beta_2$. When $\beta_2=0.8$, for example, any value of $\alpha_2\ge 0.5$ would suffice for the \aGFP mechanism, while the \aVCG mechanism would require that $\alpha_2\ge 0.75$. An illustration is provided in \figref{fig:uniform}. \figref{fig:gfp-vs-vcg-incomplete} compares the derivatives of the respective bidding functions for~$\beta_2 = 0.8$ and varying values of~$\alpha_2$.
\begin{figure}[tb]
\centering
\includegraphics{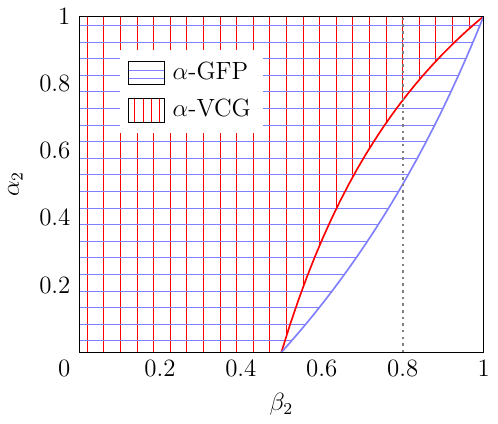}
\caption{Comparison of the \aGFP and \aVCG mechanisms under incomplete information, for a setting with two positions, three bidders, and valuations drawn independently and uniformly from~$[0,1]$. The hatched areas indicate the combinations of $\alpha_2$ and $\beta_2$ for which the mechanisms respectively possess an efficient equilibrium, when $\alpha_1=\alpha_2=1$. The dotted line at $\beta_2 = 0.8$ intersects the curve for the \aGFP mechanism at $\alpha_2 = 0.5$ and that for the \aVCG mechanism at $\alpha_2=0.75$. For all points between the intersection points the \aGFP mechanism has an efficient equilibrium and the \aVCG mechanism does not.}
\label{fig:uniform}
\end{figure}
\begin{figure}[tb]
\includegraphics{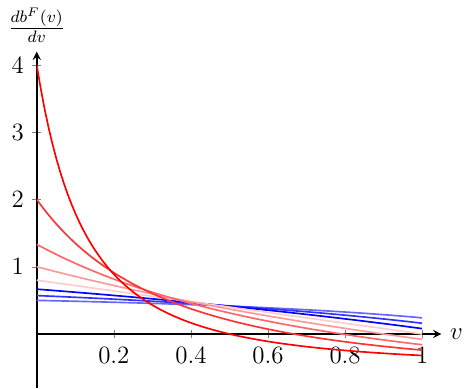}
\hspace*{\fill}
\includegraphics{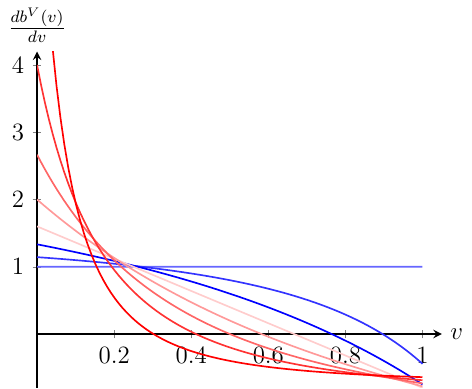}
\caption{Derivatives of the candidate bidding functions for the \aGFP and \aVCG mechanisms in a setting with three bidders with values distributed uniformly on $[0,1]$ and two positions with $\beta=(1,0.8)$, when $\alpha_1=1$ and $\alpha_2$ ranges from $0.8$ to $0.1$.} 
\label{fig:gfp-vs-vcg-incomplete}
\end{figure}

\subsection{The General Case}

Superiority of the non-truthful mechanism in preserving an efficient equilibrium again holds in general. We proceed to establish a weak superiority for any number of positions and bidders and arbitrary symmetric valuation distributions, and note that examples showing a strict separation are straightforward to construct and have indeed been given for a specific setting.
\begin{theorem}  \label{thm:incomplete}
	Let $\alpha,\beta\in\R_{\ge}^k$. Let $v\in\R^n$, with components drawn independently from a continuous distribution with bounded support. Then the \aGFP mechanism possesses an efficient Bayes-Nash equilibrium for valuations given by~$\beta$ and~$v$ whenever the \aVCG mechanism does.
\end{theorem}
To obtain this general result we will follow the same basic strategy as in the special case, but will have to overcome two major difficulties on the way.

The first difficulty concerns the equilibrium bidding function for the \aVCG mechanism. Whereas deriving a bidding function for the \aGFP mechanism remains relatively straightforward even for an arbitrary number of positions and arbitrary valuation distributions, the situation becomes significantly more complex for the \aVCG mechanism due to the dependence of its payment rule on the bids for all lower positions. Specifically, when equating the two expressions for the expected payment in equilibrium, \eqref{eq:ex-eff} and \eqref{eq:ex-vcg} in the special case, and taking derivatives on both sides, the integrand in the latter no longer depends only on~$t$, the variable of integration. Instead, the conditional densities of the values of bidders assigned lower positions introduce a dependence on~$v$. When taking the derivative one would expect to obtain a differential equation, and a closed form solution to this differential equation would be required to continue with the rest of the argument. We take a different route and use a rather surprising combinatorial equivalence to obtain an alternative expression for the expected payment that only depends on~$t$.

A second difficulty arises when trying to show that~$b^F$ is increasing for a wider range of values of~$\alpha$ and~$\beta$ than~$b^V$. In the special case we could argue directly about the derivatives of the bidding functions, but this type of argument becomes infeasible rather quickly when increasing the number of positions or moving to general value distributions. The key insight that will allow us to generalize the result is that there exist functions $A:\R\rightarrow\R$ and $B:\R\rightarrow\R$ such that $b^F(v)=A(v)/B(v)$ and $b^V(v)=A'(v)/B'(v)$, where~$A'$ and~$B'$ respectively denote the derivatives of~$A$ and~$B$. This relationship is easily verified for~\eqref{eq:ex-eff} and~\eqref{eq:ex-vcg} but continues to hold in general. We use it to show that at the minimum value of~$v$ for which~$db^F(v)/dv$ is non-positive, should such a value exist,~$db^V(v)/dv$ is non-positive as well.

We begin by deriving candidate equilibrium bidding functions for the two mechanisms. Due to the more complicated structure of the payments, the case of the \aVCG mechanism is significantly more challenging.

\begin{lemma}  \label{lem:gfp}
	Let $\alpha,\beta\in\R_{\ge}^k$ with $\alpha_k>0$ and $\beta_k>0$. Suppose valuations are drawn from a distribution with support $[0,\bar{v}]$, probability density function $f$, and cumulative distribution function $F$. Then, an efficient equilibrium of the \aGFP mechanism must use a bidding function $b^F$ with
	\[
		b^F(v) = \frac{\sum_{s=1}^{k} \beta_s \int_{0}^{v}\! \frac{d P_s(t)}{dt} t \ dt}{\sum_{s=1}^{k} \alpha_s P_s(v)}.
	\]
	If $b^F$ is increasing almost everywhere, it constitutes the unique efficient equilibrium. Otherwise no efficient equilibrium exists.
\end{lemma}

\begin{lemma}  \label{lem:vcg}
	Let $\alpha,\beta\in\R_{\ge}^k$ with $\alpha_k>0$ and $\beta_k>0$. Suppose valuations are drawn from a distribution with support $[0,\bar{v}]$, probability density function $f$, and cumulative distribution function $F$. Then, an efficient equilibrium of the \aVCG mechanism must use a bidding function $b^V$ with
	\[
		b^V(v) = \frac{\sum_{s=1}^{k} \beta_s \frac{d P_s(v)}{dv} v}{\sum_{s=1}^{k} \alpha_s \frac{d P_s(v)}{dv}}.
	\]
	If $b^V$ is increasing almost everywhere, it constitutes the unique efficient equilibrium. Otherwise no efficient equilibrium exists.
\end{lemma}

Even with the candidate bidding functions $b^F$ and $b^V$ in hand, the cases where the \aGFP and \aVCG mechanisms respectively admit an efficient equilibrium seem difficult to compare. What will ultimately drive the proof of \thmref{thm:incomplete} is a rather curious relationship between the two bidding functions that is straightforward to verify given \lemref{lem:gfp} and \lemref{lem:vcg}: the numerator of $b^V$ is equal to the derivative of the numerator of $b^F$, and the denominator of $b^V$ is equal to the derivative of the denominator of $b^F$.
\begin{corollary}  \label{cor:ab}
	Let $b^F:\R\rightarrow\R$ and $b^V:\R\rightarrow\R$ be the candidate equilibrium bidding functions for the \aGFP and \aVCG mechanisms as defined in \lemref{lem:gfp} and \lemref{lem:vcg}. Then 
	\[
		b^F(v) = \frac{A(v)}{B(v)} \quad\text{and}\quad
		b^V(v) = \frac{A'(v)}{B'(v)},
	\]
	where $A(v)=\sum_{s=1}^{k}\beta_s\int_{0}^{v}\!\frac{d P_s(t)}{dt}t\;dt$ and $B(v)=\sum_{s=1}^{k}\alpha_s P_s(v)$.
\end{corollary}

\begin{figure}[tb]
\centering
\includegraphics{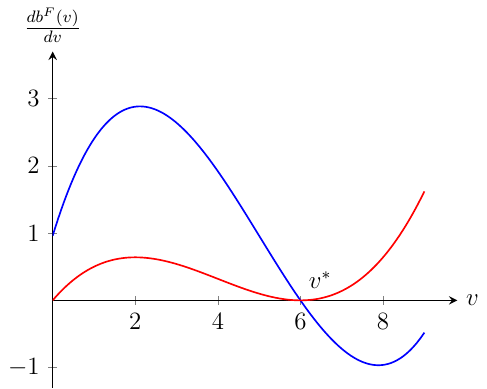}
\caption{Possible forms of the derivative of the candidate equilibrium bidding function $b^F$ when the \aGFP mechanism does not possess an equilibrium. Both the derivative and the second derivative are non-negative at zero, so if the former is non-positive anywhere on $(0,\bar{v}]$ there must be a value $v^*>0$ where it either touches or cuts the $x$-axis from above.}  \label{fig:derivative}
\end{figure}
To show that the \aGFP mechanism possesses an efficient equilibrium whenever the \aVCG mechanism does we recall that equilibrium existence is equivalent to a bidding function that is increasing almost everywhere. We first consider the candidate bidding function for the \aGFP mechanism and show that at $v=0$, either its derivative is positive or both its derivative and second derivative are non-negative. Failure to possess an equilibrium thus implies existence of a value $v^*>0$ where the derivative cuts the $x$-axis from above, or of a set of such values with positive measure where it touches the $x$-axis. In a second step we then show that the candidate bidding function for the \aVCG mechanism behaves roughly in the same way at these values. The situation is illustrated in \figref{fig:derivative}.

\begin{lemma}  \label{lem:aux1}
	Let $b^F:\R\rightarrow\R$ be the candidate equilibrium bidding function for the \aGFP mechanisms as defined in \lemref{lem:gfp}. Then,
\[
	 \left.\frac{db^F(v)}{dv}\right|_{\mathrlap{v=0}}\;\;  = \frac{n-k}{n-k+1} \cdot \frac{\beta_k}{\alpha_k}.
\]
\end{lemma}

\begin{lemma}\label{lem:aux2}
Let $b^F:\R\rightarrow\R$ be the candidate equilibrium bidding function for the \aGFP mechanisms as defined in \lemref{lem:gfp}. Then, for $n=k$,
\[
	 \left.\frac{d^2b^F(v)}{dv^2}\right|_{\mathrlap{v=0}}\;\;\ge 0.
\]
\end{lemma}

\begin{proof}[Proof of \thmref{thm:incomplete}]
	Assume that the \aGFP mechanism does not possess an equilibrium, and recall that this implies the existence of a set of values with positive measure where the derivative of~$b^F$ is not strictly increasing. By Lemmas~\ref{lem:aux1} and~\ref{lem:aux2}, there must thus exist a set of values $v^*>0$ with positive measure where
\begin{equation*}
	\left.\frac{db^F(v)}{dv}\right|_{\mathrlap{v=v^*}}\;\; = 0
	\qquad\text{and}\qquad
	\left.\frac{d^2b^F(v)}{dv^2}\right|_{\mathrlap{v=v^*}}\;\; \leq 0 ,
\end{equation*}
or one such value where the equality holds and the inequality is strict.

For an arbitrary value~$v$,
\begin{gather}
\begin{aligned}
	\frac{db^F(v)}{dv} &= \frac{A'(v)B(v)-B'(v)A(v)}{(B(v))^2}=0 \notag \\
\end{aligned}
	\intertext{requires that}
	A'(v)B(v)-B'(v)A(v)=0. \label{eq:condone}
\end{gather}
Assuming~\eqref{eq:condone},
\begin{gather}
	\begin{aligned}
	\frac{d^2b^F(v)}{dv^2} 
		&= \frac{A''(v)B(v)  - B''(v)A(v)}{(B(v))^2} \le 0 \notag\\
	\end{aligned}
	\intertext{requires that}
	A''(v)B(v)  - B''(v)A(v) \le 0, \label{eq:condtwo}
\end{gather}
Consider any $v^*>0$, and observe that $A(v^*)>0$ and $A'(v^*)>0$. For $v=v^*$ we can thus rewrite~\eqref{eq:condone} as $B(v^*)=\frac{B'(v^*)A(v^*)}{A'(v)}$, and substitute this into~\eqref{eq:condtwo} to obtain
\begin{align*}
	& A''(v^*) \frac{B'(v^*)A(v^*)}{A'(v^*)}  - B''(v^*)A(v^*) \le 0.
\end{align*}
Dividing by $A(v^*)>0$ and multiplying by $A'(v^*)>0$ yields
\begin{align*}
		& A''(v^*) B'(v^*) - A'(v^*)B''(v^*) \le 0,
\end{align*}
and thus
\begin{align*}
	& \left.\frac{b^V(v)}{dv}\right|_{\mathrlap{v=v^*}}\;\; = \frac{A''(v^*) B'(v^*) - A'(v^*)B''(v^*)}{(B'(v^*))^2} \le 0.
\end{align*}
It is, moreover, easily verified that the inequality holds strictly when $d^2b^F(v)/dv^2|_{v=v^*}<0$. There thus exists a set of values $v^*$ with positive measure where $\frac{b^V(v)}{dv}\leq 0$, and the claim follows.
\end{proof}

\appendix

\section{Complete Information: Three Positions and Four Bidders}
\label{app:three_four}

Assume without loss of generality that $v_1\geq v_2\geq v_3 \geq v_4 > 0$, and in addition that $\beta_1>\beta_2>\beta_3>0$. 
Efficiency then requires that
\begin{align}
	b_1 \geq b_2 \geq b_3 \geq b_4. \label{eq:efficiency}
\end{align}
For~$b$ to be an equilibrium, none of the bidders may benefit from raising or lowering their respective bid and being assigned a different position, which for the \aVCG mechanism means that
\begingroup\allowdisplaybreaks
\begin{align}
	\beta_1 v_1 - (\alpha_1-\alpha_2) b_2 - (\alpha_2-\alpha_3) b_3 - \alpha_3 b_4 &\geq \beta_2 v_1 - (\alpha_2-\alpha_3) b_3 - \alpha_3 b_4, \label{eq:vcg12a} \\
	\beta_1 v_1 - (\alpha_1-\alpha_2) b_2 - (\alpha_2-\alpha_3) b_3 - \alpha_3 b_4 &\geq \beta_3 v_1 - \alpha_3 b_4, \label{eq:vcg13a} \\	
	\beta_1 v_1 - (\alpha_1-\alpha_2) b_2 - (\alpha_2-\alpha_3) b_3 - \alpha_3 b_4 &\geq 0, \label{eq:vcg14a} \\
	\beta_2 v_2 - (\alpha_2-\alpha_3) b_3 - \alpha_3 b_4 &\geq \beta_1 v_2 - (\alpha_1-\alpha_2) b_1 - (\alpha_2-\alpha_3) b_3 - \alpha_3 b_4, \label{eq:vcg21a} \\
	\beta_2 v_2 - (\alpha_2-\alpha_3) b_3 - \alpha_3 b_4 &\geq \beta_3 v_2 -  \alpha_3 b_4, \label{eq:vcg23a} \\
	\beta_2 v_2 - (\alpha_2-\alpha_3) b_3 - \alpha_3 b_4 &\geq 0, \label{eq:vcg24a} \\
	\beta_3 v_3 - \alpha_3 b_4 &\geq \beta_1 v_3 - (\alpha_1-\alpha_2) b_1 - (\alpha_2-\alpha_3) b_2 - \alpha_3 b_4, \label{eq:vcg31a} \\
	\beta_3 v_3 - \alpha_3 b_4  &\geq \beta_2 v_3 - (\alpha_2-\alpha_3) b_2 - \alpha_3 b_4, \label{eq:vcg32a} \\
	\beta_3 v_3 - \alpha_3 b_4  &\geq 0, \label{eq:vcg34a} \\
	0 &\geq \beta_1 v_4 - (\alpha_1-\alpha_2) b_1 - (\alpha_2-\alpha_3) b_2 - \alpha_3 b_3, \label{eq:vcg41a} \\
	0 &\geq \beta_2 v_4 - (\alpha_2-\alpha_3) b_2 - \alpha_3 b_3, \label{eq:vcg42a} \\
	0 &\geq \beta_3 v_4 - \alpha_3 b_3. \label{eq:vcg43a}
\end{align}
\endgroup
By~\eqref{eq:vcg21a}, $(\alpha_1-\alpha_2)b_1\geq(\beta_1-\beta_2)v_2$ and thus $\alpha_1>\alpha_2$. By~\eqref{eq:vcg32a}, $(\alpha_2-\alpha_3)b_2\geq(\beta_2-\beta_3)v_3$ and thus $\alpha_2>\alpha_3$. By~\eqref{eq:vcg43a}, $\alpha_3 b_3\geq\beta_3 v_4$ and thus $\alpha_3>0$.
There are no upper bounds on~$b_1$ and no lower bounds on $b_4$ except $b_4\geq 0$, and setting~$b_1$ to a large value and $b_4=0$ satisfies \eqref{eq:vcg21a}, \eqref{eq:vcg31a}, \eqref{eq:vcg34a}, and \eqref{eq:vcg41a}.
It is furthermore easy to see that \eqref{eq:vcg14a} is implied by \eqref{eq:vcg13a} and that \eqref{eq:vcg24a} is implied by \eqref{eq:vcg23a}.
Since $(\beta_1-\beta_3)v_1-(\alpha_2-\alpha_3)b_3\geq (\beta_1-\beta_3)v_1-(\beta_2-\beta_3)v_2\geq (\beta_1-\beta_3)v_1-(\beta_2-\beta_3)v_1=(\beta_1-\beta_2)v_1$, where the first inequality holds because, by~\eqref{eq:vcg23a}, $(\alpha_2-\alpha_3)b_3\leq(\beta_2-\beta_3)v_2$, and the second inequality because $v_1\geq v_2$, \eqref{eq:vcg13a} is implied by \eqref{eq:vcg12a}.
Since $\beta_2v_4-\alpha_3b_3\leq\beta_2v_4-\beta_3v_4=(\beta_2-\beta_3)v_4\leq(\beta_2-\beta_3)v_3$, where the first inequality holds because, by~\eqref{eq:vcg43a}, $\alpha_3b_3\geq\beta_3v_4$, and the second inequality because $v_3\geq v_4$, \eqref{eq:vcg42a} is implied by \eqref{eq:vcg32a}.
Since $\alpha_1-\alpha_2>0$, $\alpha_2-\alpha_3>0$, and $\alpha_3>0$, we can rewrite the remaining constraints \eqref{eq:vcg12a}, \eqref{eq:vcg23a}, \eqref{eq:vcg32a}, and \eqref{eq:vcg43a} as upper and lower bounds on~$b_3$ and~$b_4$ and conclude that the \aVCG mechanism possesses an \emph{efficient} equilibrium if and only if there exist bids $b_2$ and $b_3$ such that
\begin{equation}  \label{eq:vcgbx}
\begin{aligned}
	\frac{(\beta_1 - \beta_2) v_1}{\alpha_1-\alpha_2} &\geq b_2 \geq \max\biggl\{\frac{(\beta_2 - \beta_3) v_3}{\alpha_2-\alpha_3},b_3\biggr\}, \\ 
	\frac{(\beta_2 - \beta_3) v_2}{\alpha_2-\alpha_3} &\geq b_3 \geq \frac{\beta_3 v_4}{\alpha_3}. 
\end{aligned}
\end{equation}

Analogously, for the \aGSP mechanism, the equilibrium conditions require that
\begin{align}
	\beta_1 v_1 - \alpha_1 b_2 &\geq \beta_2 v_1 - \alpha_2 b_3, \label{eq:gsp12a} \\
	\beta_1 v_1 - \alpha_1 b_2 &\geq \beta_3 v_1 - \alpha_3 b_4,  \label{eq:gsp13a} \\
	\beta_1 v_1 - \alpha_1 b_2 &\geq 0, \label{eq:gsp14a} \\
	\beta_2 v_2 - \alpha_2 b_3 &\geq \beta_1 v_2 - \alpha_1 b_1, \label{eq:gsp21a} \\
	\beta_2 v_2 - \alpha_2 b_3 &\geq \beta_3 v_2 - \alpha_3 b_4, \label{eq:gsp23a}\\
	\beta_2 v_2 - \alpha_2 b_3 &\geq 0, \label{eq:gsp24a} \\
	\beta_3 v_3 - \alpha_3 b_4  &\geq \beta_1 v_3 - \alpha_1 b_1, \label{eq:gsp31a} \\
	\beta_3 v_3 - \alpha_3 b_4 &\geq \beta_2 v_3 - \alpha_2 b_2, \label{eq:gsp32a} \\
	\beta_3 v_3 - \alpha_3 b_4 &\geq 0, \label{eq:gsp34a} \\
	0 &\geq \beta_1 v_4 - \alpha_1 b_1, \label{eq:gsp41a} \\
	0 &\geq \beta_2 v_4 - \alpha_2 b_2, \label{eq:gsp42a} \\
	0 &\geq \beta_3 v_4 - \alpha_3 b_3. \label{eq:gsp43a}
\end{align}
For \eqref{eq:gsp41a}, \eqref{eq:gsp42a}, \eqref{eq:gsp43a} it must be the case that $\alpha_1>0$, $\alpha_2>0$, and $\alpha_3>0$, which is weaker than the corresponding condition for the \aVCG mechanism.
There are again no upper bounds on~$b_1$, and setting it to a large value satisfies \eqref{eq:gsp21a}, \eqref{eq:gsp31a}, and \eqref{eq:gsp41a}.
Since $\beta_3 v_1-\alpha_3 b_4\geq\beta_3 v_2-\alpha_3 b_4\geq\beta_3 v_3-\alpha_3 b_4\geq 0$, where the first two inequalities hold because $v_1\geq v_2\geq v_3$ and the third inequality by~\eqref{eq:gsp34a}, \eqref{eq:gsp14a} is implied by \eqref{eq:gsp13a} and \eqref{eq:gsp24a} by \eqref{eq:gsp23a}.
Since $\beta_2v_1-\alpha_2b_3\geq\beta_2v_1-(\beta_2-\beta_3)v_2-\alpha_3b_4\geq \beta_2v_1-(\beta_2-\beta_3)v_1-\alpha_3b_4=\beta_3v_1-\alpha_3b_4$, where the first inequality holds because, by~\eqref{eq:gsp23a}, $\alpha_2b_3\leq(\beta_2-\beta_3)v_2+\alpha_3b_4$, and the second inequality because $v_1\geq v_2$, \eqref{eq:gsp13a} is implied by \eqref{eq:gsp12a}.
Since $\alpha_1>0$, $\alpha_2>0$, and $\alpha_3>0$, we can rewrite the remaining constraints \eqref{eq:gsp12a}, \eqref{eq:gsp23a}, \eqref{eq:gsp32a}, \eqref{eq:gsp34a}, \eqref{eq:gsp42a}, and \eqref{eq:gsp43a} as upper and lower bounds on $b_2$, $b_3$, and $b_4$ and conclude that the \aGSP mechanism possesses an \emph{efficient} equilibrium if and only if there exist bids $b_2$, $b_3$, and $b_4$ such that
\begin{equation}  \label{eq:gspbx}
\begin{aligned}
	\frac{(\beta_1-\beta_2) v_1 + \alpha_2 b_3}{\alpha_1} &\geq b_2 \geq
	\max \biggl\{ \frac{(\beta_2 - \beta_3) v_3 + \alpha_3 b_4}{\alpha_2} , \frac{\beta_2 v_4}{\alpha_2}, b_3 \biggr\}, \\ 
	\frac{(\beta_2 - \beta_3) v_2 + \alpha_3 b_4}{\alpha_2} &\geq b_3 \geq \max \biggl\{\frac{\beta_3 v_4}{\alpha_3}, b_4 \biggr\}, \\ 
	\frac{\beta_3 v_3}{\alpha_3} &\geq b_4. 
\end{aligned}
\end{equation}
It is not immediately obvious when these constraints can be satisfied, and why they should in fact be easier to satisfy than the constraints for the \aVCG mechanism. That~$b_2$ and~$b_3$ are each subject to more than one lower bound, and that~$b_4$ affects both the lower bound on~$b_2$ and the upper bound on~$b_3$, seems particularly unpleasant.

In \secref{sec:complete-general} we established that, even in the general case with an arbitrary number of positions and bidders, the \aGSP mechanism possesses an efficient Nash equilibrium whenever the \aVCG mechanism does. This is achieved by considering a particular, maximal, solution to the constraints for the \aVCG mechanism and mapping it to a solution to the constraints for the \aGSP mechanism. 
Instead of repeating the argument here, we show strict superiority of the \aGSP mechanism by focusing on the case where $\beta_1=\alpha_1=1$, $\beta_3=\alpha_3$, and $v_3=v_4$. By specializing~\eqref{eq:vcgbx} and~\eqref{eq:gspbx}, which requires some work and in the case of the \aVCG mechanism involves showing that one of the resulting lower bounds is always stronger than the other, we see that the \aVCG mechanism possesses an efficient equilibrium if and only if
\begin{align*}
	\frac{\beta_3(1-\beta_2)v_1+(\beta_2-\beta_3)v_3}{(1-\beta_2)v_1+(\beta_2-\beta_3)v_3} 
	&\leq \alpha_2 \leq
	\frac{(\beta_2-\beta_3)v_2+\beta_3v_3}{v_3}, 
\intertext{and the \aGSP mechanism if and only if}
	\frac{\beta_2v_3}{(1-\beta_2)v_1+\beta_2v_3} 
	&\leq \alpha_2 \leq
	\frac{(\beta_2-\beta_3)v_2+\beta_3v_3}{v_3}. 
\end{align*}
The upper bounds are identical in both cases, and it is not difficult to see that the lower bound for the \aGSP mechanism is easier to satisfy than that for the \aVCG mechanism. We compare the bounds in \figref{fig:complete34}, and note that existence of an efficient equilibrium may fail due to over- as well as underestimation of the relative values of the positions.
\begin{figure}[tb]
\centering
\includegraphics{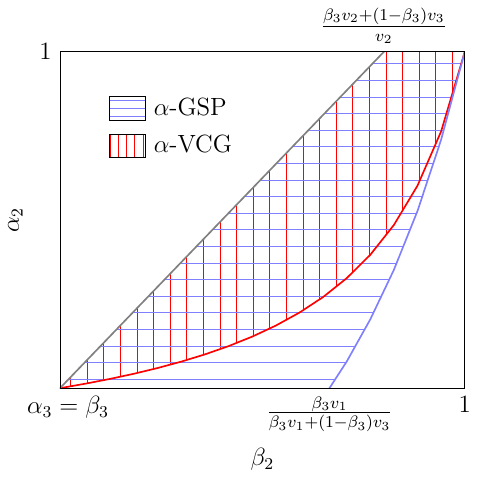}
 \caption{Comparison of the \aGSP and \aVCG mechanisms under complete information, for a setting with three positions and four bidders where $\beta_1=\alpha_1=1$, $\beta_3=\alpha_3$, and $v_3=v_4$. The hatched areas indicate the combinations of~$\alpha_2$ and~$\beta_2$ for which the mechanisms respectively possess an efficient equilibrium. The common upper bound on these areas always starts at the origin and reaches $\alpha_2=1$ at $\beta_2=(\beta_3v_2+(1-\beta_3)v_3)/v_2$. The lower bounds for both mechanisms end at the top-right corner. That for the \aGSP mechanism meets the horizontal axis at $\beta_2=\beta_3v_1/(\beta_3v_1+(1-\beta_3)v_3)$, whereas that for the \aVCG mechanism starts at the origin and curves more strongly toward the bottom-right corner as~$v_3$ decreases.}
	\label{fig:complete34}
\end{figure}

\section{Proof of~\lemref{lem:gfp}}

Since efficient equilibria must be symmetric, we can write an efficient equilibrium of the \aGFP mechanism in terms of a bidding function $b^F:[0,\bar{v}]\rightarrow\R_{\ge 0}$. A bidder with value~$v$ who is allocated position~$s$ then pays $\alpha_s b^F(v)$, and we have that
\begin{align}  \label{eq:gfp}
	\E\left[p^F(v)\right] = \sum_{s=1}^{k} \alpha_s P_s(v) b^F(v).
\end{align}
The expected payment in an efficient equilibrium is given by \lemref{lem:myerson}, and by equating~\eqref{eq:gfp} with \eqref{eq:eff} and solving for $b^F(v)$ we obtain
\[
	b^F(v) = \frac{\sum_{s=1}^{k} \beta_s \int_{0}^{v}\! \frac{d P_s(t)}{dt} t\, dt}{\sum_{s=1}^{k} \alpha_s P_s(v)}.
\]
Bidding below $b^{F}(0)=0$ is impossible and bidding above $b^{F}(\bar{v})$ is dominated, so the claim follows from \lemref{lem:myerson}.

\section{Proof of~\lemref{lem:vcg}}

	Efficiency again requires symmetry, so any efficient equilibrium of the \aVCG mechanism can be described by a bidding function $b^V:[0,\bar{v}]\rightarrow\R_{\ge 0}$. 
	
	Denote by $p^V(v)$ the payment in the \aVCG mechanism of a bidder with value $v$, and by $p_s^V(v)$ the same payment under the condition that the bidder has the $s$-highest value overall. These quantities are random variables that depend on the values of $n-1$ other bidders, and we have that
	\begin{align}  \label{eq:expay}
		\E[p^V(v)] &= \sum_{s=1}^k P_s(v)\cdot \E[p_s^V(v)],
	\end{align}
	where, as before, $P_s(v)$ is the probability that~$v$ is the $s$-highest of $n$ values drawn independently from~$F$.	
	The conditional payment $p_s^V(v)$ depends on the conditional densities of the valuations of bidders assigned lower positions, and on their resulting bids. For $s\in\{1,\dots,k\}$ and $\ell\in\{s,\dots,k\}$, denote by
	\begin{align*}
		I_{\ell,s}(v,t) &= \frac{(n-s)f(t)\binom{n-s-1}{n-\ell-1}F(t)^{n-\ell-1}(F(v)-F(t))^{\ell-s}}{F(v)^{n-s}}
	\end{align*}
	the density at~$t$ of the $(\ell+1)$-highest of $n$ values drawn independently from~$F$, under the condition that the $s$-highest value is equal to~$v$. Then
	\begin{align*}
		\E[p_s^V(v)] &= \sum_{\ell = s}^{k} (\alpha_\ell - \alpha_{\ell+1})\cdot \int_{0}^{v} I_{\ell,s}(v,t)\, b^V(t) \;dt,
	\end{align*}
	and by substituting into~\eqref{eq:expay}, exchanging the order of summation and integration, and grouping by coefficients of~$\alpha_s$, we obtain
	\begin{align}  \label{eq:payment}
		\E[p^V(v)] &= \sum_{s=1}^k P_s(v) \sum_{\ell=s}^k (\alpha_\ell-\alpha_{\ell+1}) \int_0^v I_{\ell,s}(v,t)\, b^V(t)\;dt \notag \\
		&= \int_0^v \sum_{s=1}^k \alpha_s \bigg[\sum_{\ell=1}^s P_\ell(v)\cdot I_{s,\ell}(v,t) - 
		\sum_{\ell=1}^{s-1} P_\ell(v)\cdot I_{s-1,\ell}(v,t)\bigg] b^V(t)\;dt. 
	\end{align}
	Note that the roles of~$s$ and~$\ell$ have been reversed, such that $s\geq\ell$ henceforth.
	We now recall that
	\begin{align*}
		P_\ell(v) &= \binom{n-1}{\ell-1}(1-F(v))^{\ell-1}F(v)^{n-\ell}
	\end{align*}
	and consider each of the two sums inside the parentheses in turn. 
	
	Denoting
	\begin{align*}
		J_{\ell,s} = \binom{n-1}{\ell-1} \binom{n-\ell-1}{n-s-1} (n-\ell),
	\end{align*}
	we have that
	\begin{align*}
		\sum_{\ell=1}^s P_\ell(v)\cdot I_{s,\ell}(v,t) &=
		\sum_{\ell=1}^s J_{\ell,s}\cdot (1-F(v))^{\ell-1} f(t) F(t)^{n-s-1} (F(v)-F(t))^{s-\ell} \\
		&= \sum_{\clapstack{1\leq\ell\leq s\\0\leq x\leq\ell-1\\0\leq y\leq s-\ell}} J_{\ell,s}
		 \binom{\ell-1}{x} \binom{s-\ell}{y} (-1)^{\ell+y-x-1} f(t) F(v)^{s-x-y-1} F(t)^{n+y-s-1},
	\end{align*}
	where the second equality holds because by the binomial theorem
	\begin{align*}
		(1-F(v))^{\ell-1} &= \sum_{x=0}^{\ell-1} \binom{\ell-1}{x} (-F(v))^{\ell-x-1}
		\quad\text{and} \\
		(F(v)-F(t))^{s-\ell} &= \sum_{y=0}^{s-\ell} \binom{s-\ell}{y} F(v)^{s-\ell-y} (-F(t))^y.
	\end{align*}
	We claim that the terms with $x+y<s-1$ cancel out, \ie that
	\begin{align*}
		\sum_{\clapstack{1\leq\ell\leq s\\0\leq x\leq\ell-1\\0\leq y\leq s-\ell\\x+y<s-1}} J_{\ell,s}
		  & \binom{\ell-1}{x} \binom{s-\ell}{y} (-1)^{\ell+y-x-1} f(t) F(v)^{s-x-y-1} F(t)^{n+y-s-1} \\[-4ex]
		&= \sum_{\clapstack{0\leq z\leq s-2\\0\leq y\leq z\\z-y+1\leq\ell\leq s-y}} J_{\ell,s} \binom{\ell-1}{z-y} \binom{s-\ell}{y} (-1)^{\ell+2y-z-1} F(v)^{s-z-1} f(t) F(t)^{n+y-s-1} = 0.
	\end{align*}
	Indeed, the first equality follows by setting $z=x+y$ and observing that in both sums $\ell$ takes exactly the values between $x+1=z-y+1$ and $s-y$. The second equality holds because for any $z$ and $y$ with $0\leq z\leq s-2$ and~$0\leq y\leq z$,
	\begingroup\allowdisplaybreaks
	\begin{align*}
		\sum_{\mathclap{\ell=z-y+1}}^{s-y} J_{\ell,s} & \binom{\ell-1}{z-y} \binom{s-\ell}{y} (-1)^{\ell+2y-z-1} \\
		&= \sum_{\mathllap{\ell=z-y+1}}^{s-y} \binom{n-1}{\ell-1} \binom{n-\ell-1}{n-s-1} (n-\ell) \binom{\ell-1}{z-y} \binom{s-\ell}{y} (-1)^{\ell+2y-z-1} \\
		&= \frac{(n-1)!}{(n-s-1)!(z-y)!y!} \sum_{\ell=z-y+1}^{s-y} \frac{(-1)^{\ell+2y-z-1}}{(\ell-z+y-1)!(s-\ell-y)!} \\
		&= \frac{(n-1)!}{(n-s-1)!(z-y)!y!} \sum_{j=0}^{m} \frac{(-1)^{j+y}}{j!(m-j)!} \\
		&= \frac{(n-1)!(-1)^y}{(n-s-1)!(z-y)!y!m!} \sum_{j=0}^{m} (-1)^{j}\binom{m}{j} \\
		&= \frac{(n-1)!(-1)^y}{(n-s-1)!(z-y)!y!m!} (1+(-1))^m = 0, \tagthis\label{eq:binom}
	\end{align*}\endgroup
	where the third equality follows by setting $j=\ell-z+y-1$ and $m=s-z-1$ and the fifth equality holds by the binomial theorem.
	Thus, actually,
	\begingroup\allowdisplaybreaks
	\begin{align*}
		\sum_{\ell=1}^s P_\ell(v)\cdot I_{s,\ell}(v,t)
		&= \sum_{\clapstack{1\leq\ell\leq s\\0\leq x\leq\ell-1\\0\leq y\leq s-\ell\\x+y=s-1}} J_{\ell,s}
		  \binom{\ell-1}{x} \binom{s-\ell}{y} (-1)^{\ell+y-x-1} f(t) F(v)^{s-x-y-1} F(t)^{n+y-s-1} \\
		&= \sum_{\ell=1}^{s} J_{\ell,s} \binom{\ell-1}{\ell-1} \binom{s-\ell}{s-\ell} (-1)^{s-\ell} f(t) F(v)^0 F(t)^{n-\ell-1} \\
		&= \sum_{\ell=1}^{s}J_{\ell,s}\cdot (-1)^{s-\ell} F(t)^{n-\ell-1} f(t) \\
		&= \sum_{\ell=1}^{s}\binom{n-1}{s-1}(n-s)\binom{s-1}{\ell-1}(-1)^{s-\ell}F(t)^{n-\ell-1}f(t) \\
		&= \sum_{\ell=0}^{s-1}\binom{n-1}{s-1}(n-s)\binom{s-1}{\ell}(-1)^{s-\ell-1}F(t)^{n-\ell-2}f(t) \\
		&= \binom{n-1}{s-1} (1-F(t))^{s-1}(n-s)F(t)^{n-s-1}f(t),  \tagthis\label{eq:aequal}
	\end{align*}\endgroup
	where the third equality holds because
	\begin{align*}
		J_{\ell,s} &= \binom{n-1}{\ell-1} \binom{n-\ell-1}{n-s-1} (n-\ell)
		= \frac{(n-1)!}{(n-\ell)!(l-1)!}\; \frac{(n-\ell-1)!}{(s-\ell)!(n-s-1)!}\; (n-\ell) \\[1ex]
		&= \frac{(n-1)!}{(l-1)!(s-\ell)!(n-s-1)!}
		= \frac{(n-1)!}{(n-s)!(s-1)!}\; \frac{(s-1)!}{(s-\ell)!(\ell-1)!}\; (n-s) \\[1ex]
		&= \binom{n-1}{s-1} \binom{s-1}{\ell-1} (n-s)
	\end{align*}
	and the fifth equality because by the binomial theorem
	\[
		\sum_{\ell=0}^{s-1}\binom{s-1}{\ell}(-1)^{s-\ell-1}F(t)^{s-\ell-1} = (1-F(t))^{s-1}.
	\]	
	
	Analogously, for the second term in the parentheses of~\eqref{eq:payment},
	\begin{align*}
		\sum_{\ell=1}^{s-1} P_\ell(v) & {} \cdot {} I_{s-1,\ell}(v,t) \\[-2ex]
		&= \sum_{\ell=1}^{s-1} J_{\ell,s-1}\cdot (1-F(v))^{\ell-1} f(t) F(t)^{n-s} (F(v)-F(t))^{s-\ell-1} \\
		&= \sum_{\clapstack{1\leq\ell\leq s-1\\0\leq x\leq\ell-1\\0\leq y\leq s-\ell-1}} J_{\ell,s-1}\cdot \binom{\ell-1}{x} \binom{s-\ell-1}{y} (-1)^{\ell+y-x-1} f(t) F(v)^{s-x-y-2} F(t)^{n+y-s},
	\end{align*}
	where the second equality holds because by the binomial theorem
	\begin{align*}
		(1-F(v))^{\ell-1} &= \sum_{x=0}^{\ell-1} \binom{\ell-1}{x} (-F(v))^{\ell-x-1}
		\quad\text{and} \\
		(F(v)-F(t))^{s-\ell-1} &= \sum_{y=0}^{s-\ell-1} \binom{s-\ell-1}{y} F(v)^{s-\ell-y-1} (-F(t))^y.
	\end{align*}
	We claim that the terms with $x+y<s-2$ cancel out, \ie that
	\begin{align*}
		\sum_{\clapstack{1\leq\ell\leq s-1\\0\leq x\leq\ell-1\\0\leq y\leq s-\ell-1\\x+y<s-2}} J_{\ell,s-1}
		  & \binom{\ell-1}{x} \binom{s-\ell-1}{y} (-1)^{\ell+y-x-1} f(t) F(v)^{s-x-y-2} F(t)^{n+y-s} \\[-4ex]
			&= \sum_{\clapstack{0\leq z\leq s-3\\0\leq y\leq z\\ z-y+1\leq\ell\leq s-y-1}} J_{\ell,s-1}
			  \binom{\ell-1}{z-y} \binom{s-\ell-1}{y} (-1)^{\ell+2y-z-1} f(t) F(v)^{s-z-2} F(t)^{n+y-s} = 0.
	\end{align*}
	Indeed, the first equality follows by setting $z=x+y$ and observing that in both sums $\ell$ takes exactly the values between $x+1=z-y+1$ and $s-y-1$. The second equality holds because for any $z$ and $y$ with $0\leq z\leq s-3$ and $0\leq y\leq z$,
	\begin{equation*}
		\sum_{\mathclap{\ell=z-y+1}}^{s-y-1} J_{\ell,s-1} \binom{\ell-1}{z-y} \binom{s-\ell-1}{y} (-1)^{\ell+2y-z-1} 
		= \sum_{\mathclap{\ell=z-y+1}}^{r-y} J_{\ell,r} \binom{\ell-1}{z-y} \binom{r-\ell}{y} (-1)^{\ell+2y-z-1} = 0,
	\end{equation*}
	where the first equality follows by setting $r=s-1$ and the second equality holds by~\eqref{eq:binom}.
	Thus, actually,
	\begingroup\allowdisplaybreaks
	\begin{align*}
		\sum_{\ell=1}^{s-1} P_\ell(v)\cdot I_{s-1,\ell}(v,t) 
		&= \sum_{\clapstack{1\leq\ell\leq s-1\\0\leq x\leq\ell-1\\0\leq y\leq s-\ell-1\\x+y=s-2}} J_{\ell,s-1}
		  \binom{\ell-1}{x} \binom{s-\ell-1}{y} (-1)^{\ell+y-x-1} f(t) F(v)^{s-x-y-2} F(t)^{n+y-s} \\
		&= \sum_{\ell=1}^{s-1} J_{\ell,s-1} \binom{\ell-1}{\ell-1} \binom{s-\ell-1}{s-\ell-1} (-1)^{s-\ell-1} f(t) F(v)^0 F(t)^{n-\ell-1} \\
		&= \sum_{\ell=1}^{s-1}J_{\ell,s-1}\cdot (-1)^{s-\ell-1} F(t)^{n-\ell-1} f(t) \\
		&= \sum_{\ell=1}^{s-1}\binom{n-1}{s-1}(s-1)\binom{s-2}{\ell-1}(-1)^{s-\ell-1}F(t)^{n-\ell-1}f(t) \\
		&= \sum_{\ell=0}^{s-2}\binom{n-1}{s-1}(s-1)\binom{s-2}{\ell}(-1)^{s-\ell-2}F(t)^{n-\ell-2}f(t) \\
		&= \binom{n-1}{s-1} (1-F(t))^{s-2}(s-1)F(t)^{n-s}f(t),  \tagthis\label{eq:bequal}
	\end{align*}\endgroup
	where the third equality holds because
	\begin{align*}
		J_{\ell,s-1} &= \binom{n-1}{\ell-1} \binom{n-\ell-1}{n-s} (n-\ell)
		= \frac{(n-1)!}{(n-\ell)!(l-1)!}\; \frac{(n-\ell-1)!}{(s-\ell-1)!(n-s)!}\; (n-\ell) \\[1ex]
		&= \frac{(n-1)!}{(l-1)!(s-\ell-1)!(n-s)!}
		= \frac{(n-1)!}{(n-s)!(s-1)!}\; \frac{(s-2)!}{(s-\ell-1)!(\ell-1)!}\; (s-1) \\[1ex]
		&= \binom{n-1}{s-1} \binom{s-2}{\ell-1} (s-1)
	\end{align*}
	and the fifth equality because by the binomial theorem
	\[
		\sum_{\ell=0}^{s-2}\binom{s-2}{\ell}(-1)^{s-\ell-2}F(t)^{s-\ell-2} = (1-F(t))^{s-2}.
	\]	
	
	By substituting~\eqref{eq:aequal} and~\eqref{eq:bequal} into~\eqref{eq:payment}, we conclude that
	\begin{align}  \label{eq:vcg}
		\E\left[p^V(v)\right] &=
		\int_0^v \sum_{s=1}^k \alpha_s \biggl(
		\binom{n-1}{s-1} (1-F(t))^{s-1}(n-s)F(t)^{n-s-1}f(t) - {} \notag \\
		&\hspace*{4cm} \binom{n-1}{s-1} (1-F(t))^{s-2}(s-1)F(t)^{n-s}f(t)
		\biggr) b^V(t)\;dt \notag \\
		&= \sum_{s=1}^{k} \alpha_s \int_0^v \frac{d P_s(t)}{dt} b^V(t) \;dt.
	\end{align}
	
	The expected payment in an efficient equilibrium is again given by \lemref{lem:myerson}. We can thus equate~\eqref{eq:vcg} with~\eqref{eq:eff}, take derivatives on both sides, and solve for $b^{V}(v)$ to obtain
\[
	b^{V}(v) 
	= \frac{\sum_{s=1}^{k} \beta_s \frac{d P_s(v)}{dv} v}{\sum_{s=1}^{k} \alpha_s \frac{d P_s(v)}{dv}}.
\]
Bidding below $b^{V}(0)=0$ is impossible and bidding above $b^{V}(\bar{v})$ is dominated, so the claim follows from \lemref{lem:myerson}.

\section{Proof of~\lemref{lem:aux1}}

By \corref{cor:ab}, $b^F(v)=A(v)/B(v)$ with $A(v)=\sum_{s=1}^{k}\beta_s\int_{0}^{v}\!\frac{d P_s(t)}{dt}t\;dt$ and $B(v)=\sum_{s=1}^{k}\alpha_s P_s(v)$. Writing the derivative as a limit of difference quotients, applying l'Hospital's rule to each of the two resulting terms, and respectively substituting~$x$ for $2\delta$ and $\delta$ we obtain
\begingroup\allowdisplaybreaks
\begin{align*}
	\left.\frac{db^F(v)}{dv}\right|_{\mathrlap{v=0}}\;\; &= \lim_{\delta \rightarrow 0} \left(\frac{A(2\delta)/B(2\delta) - A(\delta)/B(\delta)}{\delta}\right) \\
	&= \lim_{\delta\rightarrow 0} \frac{A(2\delta)}{\delta\cdot B(2\delta)} - \lim_{\delta\rightarrow 0} \frac{A(\delta)}{\delta\cdot B(\delta)} \\
	&= \lim_{\delta\rightarrow 0} \frac{A'(2\delta)\cdot 2}{\delta\cdot B'(2\delta)\cdot 2 + B(2\delta)} - 
	\lim_{\delta\rightarrow 0} \frac{A'(\delta)}{\delta\cdot B'(\delta) + B(\delta)} \\
	&= \lim_{x \rightarrow 0} \frac{\left(\sum_{s=1}^{k} \beta_s \frac{d P_s(x)}{dx} \cdot x \right) \cdot 2}{\left(\sum_{s=1}^{k} \alpha_s \frac{dP_s(x)}{dx} \cdot x\right) + \left(\sum_{s=1}^{k} \alpha_s P_s(x) \right)} - {} \\
	&& \makebox[0pt][r]{$\displaystyle\lim_{x \rightarrow 0} \frac{\left(\sum_{s=1}^{k} \beta_s \frac{d P_s(x)}{dx} \cdot x \right)}{\left(\sum_{s=1}^{k} \alpha_s \frac{dP_s(x)}{dx} \cdot x\right) + \left(\sum_{s=1}^{k} \alpha_s P_s(x) \right)}$.} 
\end{align*}
\endgroup
To analyze these limits we extend by $1=(F(x)^{n-k-1} \cdot x)^{-1}/(F(x)^{n-k-1}\cdot x)^{-1}$, factor $(F(x)^{n-k-1}\cdot x)^{-1}$ into the numerator and denominator, and consider each of the terms in the numerator and denominator in turn.

Using~$\gamma$ as a placeholder for~$\alpha$ or~$\beta$ and replacing $P_s(x)$ by its definition,
\begin{align*}
\frac{\sum_{s=1}^{k} \gamma_s \cdot \frac{d P_s(x)}{dx} \cdot x}{F^{n-k-1}(x) \cdot x}
&= \sum_{s = 1}^{k} \gamma_s \bigg[ \binom{n-1}{s-1} (n-s) F^{k-s}(x) (1-F(x))^{s-1} f(x)\\
&\qquad\qquad\qquad- \binom{n-1}{s-1} (s-1) F^{k-s+1}(x) (1-F(x))^{s-2} f(x)\bigg]\\
&= \sum_{s=1}^{k} \sum_{\ell=0}^{s-1} \gamma_s  (-1)^\ell \binom{n-1}{s-1} (n-s) \binom{s-1}{\ell} F(x)^{k-s+\ell} f(x)\\
&\qquad\qquad\qquad- \sum_{s=1}^{k} \sum_{\ell=0}^{s-2} \gamma_s  (-1)^\ell \binom{n-1}{s-1} (s-1) \binom{s-2}{\ell} F(x)^{k-s+\ell+1} f(v) .
\end{align*}
Similarly,
\begin{align*}
\frac{\sum_{s=1}^{k} \alpha_s P_s(x)}{F^{n-k-1}(x)\cdot x}
&= \frac{1}{x} \cdot \sum_{s=1}^{k}\alpha_s\binom{n-1}{s-1}F(x)^{k-s+1}(1-F(x))^{s-1}\\
&=\frac{F(x)}{x} \cdot \sum_{s=1}^{k} \sum_{\ell=0}^{s-1} \alpha_s (-1)^\ell \binom{n-1}{s-1} \binom{s-1}{\ell} F(x)^{k-s+\ell}.
\end{align*}
Since $\lim_{v\rightarrow 0}F(x)^d=0$ for $d>0$, the only terms that survive in the limit are those where the exponent of $F(x)$ is zero. For $s\in\{1,\dots,k\}$ and $\ell\in\{0,\dots,s-1\}$, $k-s+\ell=0$ only if $s=k$ and $\ell=0$. For $s\in\{1,\dots,k\}$ and $\ell\in\{0,\dots,s-2\}$, $k-s+\ell-1\neq 0$. Using that $\lim_{x\rightarrow 0}F(x)/x=f(0)$, we thus obtain
\begin{align*}
	\left.\frac{db^F(v)}{dv}\right|_{\mathrlap{v=0}}\;\; &=
	\frac{\beta_k \binom{n-1}{k-1} (n-k) f(0) \cdot 2}{\alpha_k \binom{n-1}{k-1} (n-k) f(0) + \alpha_k \binom{n-1}{k-1}f(0)}
	- \frac{\beta_k \binom{n-1}{k-1} (n-k) f(0)}{\alpha_k \binom{n-1}{k-1} (n-k) f(0)+ \alpha_k \binom{n-1}{k-1}f(0)} \\
	&= \frac{2 (n-k) \beta_k}{(n-k+1) \alpha_k} - \frac{(n-k) \beta_k}{(n-k+1) \alpha_k}
	= \frac{n-k}{n-k+1} \cdot \frac{\beta_k}{\alpha_k}
\end{align*}
as claimed.

\section{Proof of~\lemref{lem:aux2}}

By \corref{cor:ab}, $b^F(v)=A(v)/B(v)$ with $A(v)=\sum_{s=1}^{k}\beta_s\int_{0}^{v}\!\frac{d P_s(t)}{dt}t\;dt$ and $B(v)=\sum_{s=1}^{k}\alpha_s P_s(v)$. For $n=k$, by \lemref{lem:aux1},
\[
 \left.\frac{db^F(v)}{dv}\right|_{\mathrlap{v=0}} \;\;=  \left.\frac{A'(v) B(v) - A(v) B'(v)}{B(v)^2}\right|_{\mathrlap{v=0}} \;\; = 0.
 \]
Since
\[
	B(0) = \sum_{s=1}^{k} \alpha_s P_s(0) \geq \alpha_k P_k(0) = \alpha_k (1-F(0))^{n-1} = \alpha_k > 0,
\]
this implies that
\[
	\big(A'(v) B(v) - A(v) B'(v)\big)\Bigr|_{\mathrlap{v=0}} \;\;  = 0.
\]
Thus
\begin{align*}
	\left.\frac{d^2b^F(v)}{dv^2}\right|_{\mathrlap{v=0}} \;\;
	&= \left.\frac{(A''(v)B(v)-A(v)B''(v))B(v)^2-(A'(v) B(v)}{B(v)^4}\right|_{\mathrlap{v=0}} - 
	\left.\frac{A(v) B'(v))2B(v)B'(v)}{B(v)^4}\right|_{\mathrlap{v=0}} \\[2ex]  
	&= \left.\frac{A''(v)B(v)-A(v)B''(v)}{B(v)^2}\right|_{\mathrlap{v=0}} \;\;.
\end{align*}
We have already seen that $B(0)>0$. Moreover, $A(0)=0$ by the definition of~$A$ and $B''(0)<\infty$ by assumption on the value distributions, so it suffices to show that
\[
	A''(0) = \left.\left( \sum_{s=1}^{k} \beta_s \frac{d^2P_s(v)}{dv^2} \cdot v + \sum_{s=1}^{k} \beta_s \frac{dP_s(v)}{dv} \right)\right|_{\mathrlap{v=0}} \;\; \geq 0. 
\]
Also by assumptions on the value distributions, $d^2P_s(v)/dv^2<\infty$ for all $v$, so the first term vanishes. The second term is
\begin{align*}
	\left.\sum_{s=1}^{k} \beta_s \frac{dP_s(v)}{dv}\right|_{\mathrlap{v=0}} \;\;
	&= \sum_{s=1}^{k} \beta_s \biggl(\binom{n-1}{s-1}(n-s)F(v)^{n-s-1} (1-F(v))^{s-1} f(v) \\
	&\qquad\qquad\qquad\qquad- \binom{n-1}{s-1}(s-1)F(v)^{n-s}(1-F(v))^{s-2} f(v) \biggr)\biggr|_{\mathrlap{v=0}} \\
	&= \beta_{k-1} (k-1) f(0) - \beta_k (k-1) f(0) \ge 0,
\end{align*}
where we have used the definition of $P_s(v)$ and the fact that the only non-zero terms are those where the exponent of $F(v)$ is zero. Since $\beta_{k-1} \ge b_{k}$ and $f(0)>0$, this shows the claim.


\end{document}